\documentclass[12pt, draftclsnofoot, onecolumn]{IEEEtran}
\ifCLASSINFOpdf
\else
\fi
\usepackage{amsmath,amssymb,amsthm}

\usepackage{authblk}
\newtheorem{theorem}{Theorem}

\usepackage{hhline}
\usepackage{graphicx}
\usepackage{epstopdf}
\usepackage{subcaption}
\usepackage{algorithm}
\usepackage{algorithmic}
\usepackage{amsmath}
\usepackage{cite}

\begin{document}

\title{The Indoor Mobile Coverage Problem Using UAVs}

\author{Hazim Shakhatreh, Abdallah Khreishah, and Issa Khalil

\thanks{Hazim Shakhatreh and Abdallah Khreishah are with the Department of
	Electrical and Computer Engineering, New Jersey Institute of Technology
	(email: \{hms35,abdallah\}@njit.edu)\\
	Issa Khalil is with Qatar Computing Research Institute (QCRI), HBKU, Doha, Qatar (email: ikhalil@hbku.edu.qa)
}
\thanks{Part of this work was presented in ICC 2017~\cite{haz2017efficient} and ICICS 2017~\cite{hazim2017}.}}

\maketitle

\begin{abstract}
Unmanned aerial vehicles (UAVs) can be used as
aerial wireless base stations when cellular networks are not operational due to natural disasters. They can also be used to supplement the ground base station in order to provide better coverage and higher data rates for the users. Prior studies on UAV-based wireless coverage typically consider
an Air-to-Ground path loss model, which assumes that the users
are outdoor and located on a 2D plane. In this paper,
we propose using UAVs to provide wireless coverage for indoor users inside a high-rise building. First, we present realistic Outdoor-Indoor path loss models and describe the tradeoff introduced by these models. Then we study the problem of efficient placement of a single UAV, where the objective is to minimize the total transmit power required to cover the entire high-rise building. The formulated problem is non-convex and is generally difficult to solve. To that end, we consider three cases of practical interest and provide efficient solutions to the formulated problem under these cases. Then we study the problem of minimizing the number of UAVs required to provide wireless coverage to high rise buildings and prove that this problem is NP-complete. Due to the intractability of the problem, we use clustering to minimize the number of UAVs required to cover the indoor users. We demonstrate through simulations that the method that clusters the building into regular structures and places the UAVs in each cluster requires 80\% more number of UAVs relative to our clustering algorithm.
\end{abstract}

\begin{IEEEkeywords}
Unmanned aerial vehicles, Outdoor-to-Indoor
path loss model, gradient descent algorithm, particle swarm optimization, $k$-means clustering.
\end{IEEEkeywords}

\IEEEpeerreviewmaketitle

\section{Introduction}

UAVs can be used to provide wireless coverage during
emergency cases where each UAV serves as an aerial wireless
base station when the cellular network goes down~\cite{bupe2015relief}. They
can also be used to supplement the ground base station in
order to provide better coverage and higher data rates for the
users~\cite{bor2016efficient}.

In order to use a UAV as an aerial wireless base station, the
authors in~\cite{al2014modeling} presented an Air-to-Ground path loss model that helped the academic researchers to formulate many important UAV-based coverage problems. The authors in~\cite{mozaffari2015drone} utilized this model to evaluate the
impact of a UAV altitude on the downlink ground coverage and to determine the optimal values for altitude which lead to maximum coverage and minimum required transmit power. In~\cite{mozaffari2016optimal}, the authors used the path loss model to propose a power-efficient deployment for UAVs under the constraint of satisfying the rate requirement for all ground users. The authors in~\cite{mozaffari2016efficient} utilized the path loss model to study the optimal deployment of multiple UAVs equipped with directional antennas, using circle packing theory. The 3D locations of the UAVs are determined in a way that the total coverage area is maximized. In~\cite{kalantari2016number}, the authors used the path loss model to find the minimum number of UAVs and their 3D locations so that all outdoor ground users are served. However, it is assumed that all users are outdoor and the location of each user can be represented by an outdoor 2D point. These assumptions limit the applicability of this model when one needs to consider indoor users.

Providing good wireless coverage for indoor users is very
important. According to Ericsson report~\cite{ericsson}, 90\% of the time people are indoor and 80\% of the mobile Internet access traffic also happens indoors~\cite{alcatel,cisco}. To guarantee wireless coverage, service providers are faced with several key challenges, including providing service to a large number of indoor users and the ping pong effect due to interference from near-by macro cells~\cite{comm,amplitic,zhang2016study}. In this paper, we propose using UAVs to provide wireless coverage for users inside a high-rise building during emergency cases and special events (such as concerts, indoor sporting events, etc.), when the cellular network service is not available or it is unable to serve all indoor users. To the best of our knowledge, this is the first work that proposes using UAVs to provide wireless coverage for indoor users. We summarize our main contributions as follows:
\begin{itemize}
\item We utilize an Outdoor-Indoor path loss model for low-SHF band (450 MHz to 6 GHz)~\cite{series2009guidelines}, certified
by ITU, and an Outdoor-Indoor path loss model for high-SHF band (over 6 GHz)~\cite{imai2016outdoor}, then we show the tradeoff introduced by these models.
\item We formulate the problem of efficient placement of a single UAV, where the objective is to minimize the total transmit power
required to cover the entire high-rise building.
\item Since the formulated problem is non-convex and is generally difficult to solve, we consider three cases of practical interest and provide efficient solutions to the formulated problem under these cases and for different operating frequencies (low-SHF and high-SHF bands). In the first case, we aim to find the minimum transmit power such that an indoor user with the maximum path loss can be covered. In the second case, we assume that the locations of indoor users are symmetric across the dimensions of each floor (such as office buildings or hotels), and propose a gradient descent algorithm for finding an efficient location of a UAV. In the third case, we assume that the locations of indoor users are uniformly distributed in each floor, and propose a particle swarm optimization algorithm to find an efficient 3D placement of a UAV that tries to minimize the total transmit power required to cover the indoor users.
 \item Due to the limited transmit power of a UAV, we formulate the problem of minimizing the number of UAVs required to provide wireless coverage to high rise building and prove that this problem is NP-complete. 
 \item Due to the intractability of the problem, we use clustering to minimize the number of UAVs required to cover indoor users. We demonstrate through simulations that the method that clusters the building into regular structures and places the UAVs in each cluster requires 80\% more number of UAVs relative to our clustering algorithm.
\end{itemize}

\section{System Model}
\label{sec:system_model}
\subsection{System Settings}
\label{subsec:system settings}
Let ($x_{UAV}$,$y_{UAV}$,$z_{UAV}$) denote the 3D location of the UAV. We assume that all users are located inside a high-rise building as shown in Figure~\ref{fig1}, and use ($x_{i}$,$y_{i}$,$z_{i}$) to denote the location of user $i$. The dimensions of the high-rise building, in the shape of a rectangular prism, are $[0,x_b]$ $\times$ $[0,y_b]$ $\times$ $[0,z_b]$. Also, let $d_{out,i}$ be the distance between the UAV and indoor user $i$, let $\theta_{i}$ be the incident angle , and let $d_{in,i}$ be the distance between the building wall and indoor user $i$.

	\subsection{Outdoor-Indoor Path Loss Models}
	\label{subsec:Outdoor-Indoor Path Loss Model}
	The Air-to-Ground path loss model presented in~\cite{al2014modeling} is not appropriate when we consider wireless coverage for indoor users, because this model assumes that all users are outdoor and located at 2D points. In this paper, we adopt the Outdoor-Indoor path loss model, certified by the ITU~\cite{series2009guidelines}, for low-SHF operating frequency. The path loss is given as follows:\\
	\begin{equation*}
	\begin{split}
	L_i=L_{F}+L_{B}+L_{I}= (w\log_{10}d_{out,i}+w\log_{10}f_{Ghz}+g_{1})\\
	+(g_{2}+g_{3}(1-\cos\theta_{i})^{2})+(g_{4}d_{in,i})
	\end{split}
	\end{equation*}
	where $L_{F}$ is the free space path loss, $L_{B}$ is the building penetration loss, and $L_{I}$ is the indoor loss. In this model, we also have $w$=20, $g_{1}$=32.4, $g_{2}$=14, $g_{3}$=15,$g_{4}$=0.5~\cite{series2009guidelines} and $f_{Ghz}$ is the carrier frequency. 
	
	In~\cite{imai2016outdoor}, the authors clarify the Outdoor-to-Indoor path loss characteristics based on the measurement for 0.8 to 37 GHz frequency band. We adopt this path loss model for high-SHF operating frequency. The path loss is given as follows:\\
	\begin{equation*}
	\begin{split}
	L_i=L_{F}+L_{B}+L_{I}=
	(\alpha_1+\alpha_2\log_{10}d_{out,i}+\alpha_3\log_{10}f_{Ghz})+~~~\\
	(\beta_1+\frac{\beta_2-\beta_1}{1+exp(-\beta_3(\theta_i-\beta_4))})+(\gamma_1d_{in,i})
	\end{split}
	\end{equation*}
	 In this model, we have
	$\alpha_1$=31.4, $\alpha_2$=20, $\alpha_3$=21.5, $\beta_1$=6.8, $\beta_2$=21.8, $\beta_3$=0.453, $\beta_4$=19.7 and $\gamma_1$=0.49.
	
	Note that there is a key tradeoff in the path loss models when the horizontal distance between the UAV and a user changes. When this horizontal distance increases, the free space path loss (i.e., $L_F$) increases as $d_{out,i}$ increases, while the building penetration loss (i.e., $L_B$) decreases as the incident angle (i.e., $\theta_i$) decreases (Figure~\ref{fig2} shows the penetration loss for high-SHF band). Similarly, when this horizontal distance decreases, the free space path loss (i.e., $L_F$) decreases as $d_{out,i}$ decreases, while the building penetration loss (i.e., $L_B$) increases as the incident angle (i.e., $\theta_i$) increases. 
	\section{Providing Wireless coverage using a single UAV}
	
	\begin{figure*}[ht]
		\begin{minipage}[b]{0.25\linewidth}
			\centering
			\includegraphics[width=\textwidth]{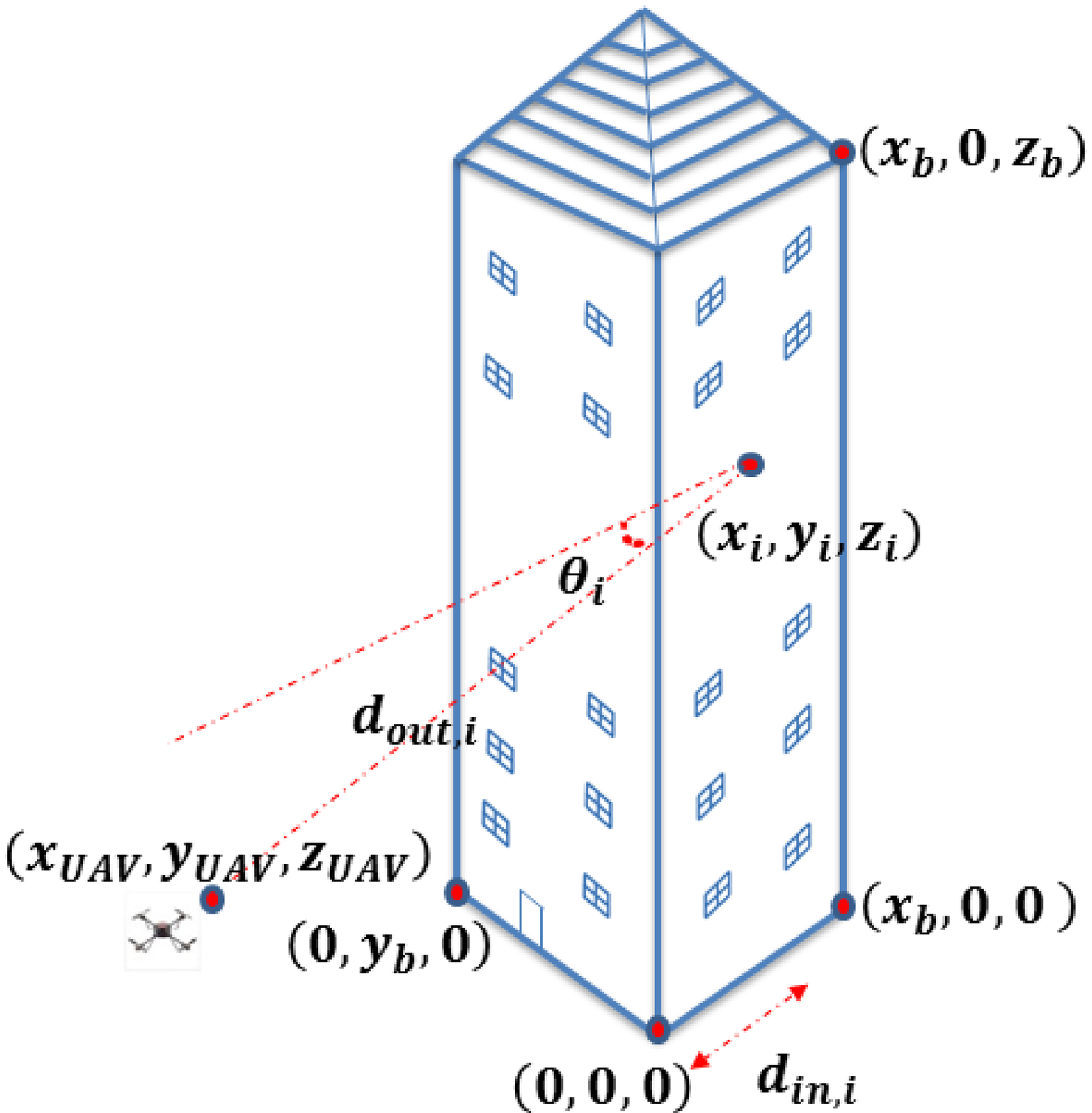}
			\caption{Parameters of the            
				path loss model
			}
			\label{fig1}
		\end{minipage}
		\hspace{0.1cm}
		\begin{minipage}[b]{0.335\linewidth}
			\centering
			\includegraphics[width=\textwidth]{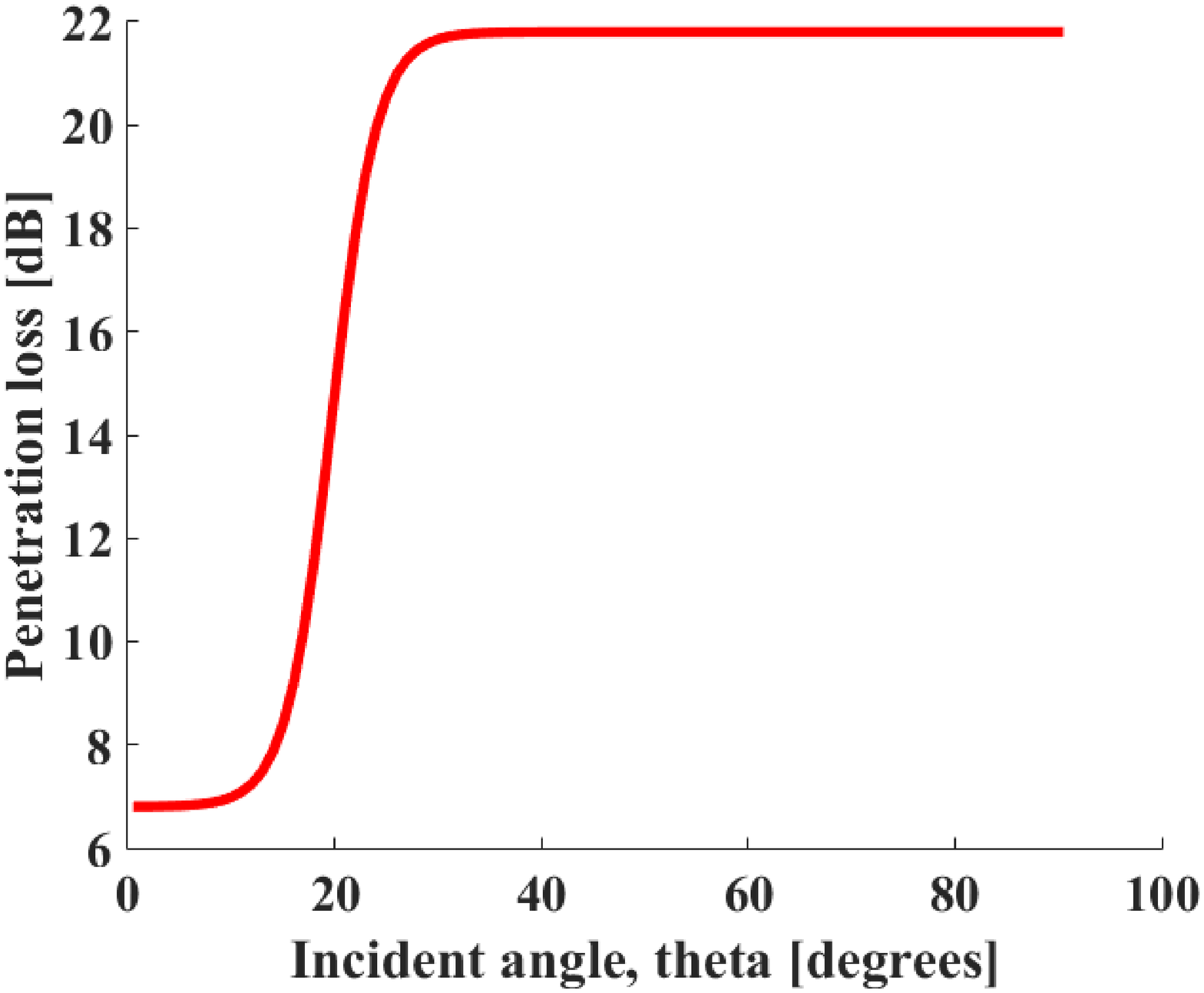}
			\caption{Building penetration loss for high-SHF}
			\label{fig2}
		\end{minipage}
		\hspace{0.1cm}
		\begin{minipage}[b]{0.355\linewidth}
			\centering
			\includegraphics[width=\textwidth]{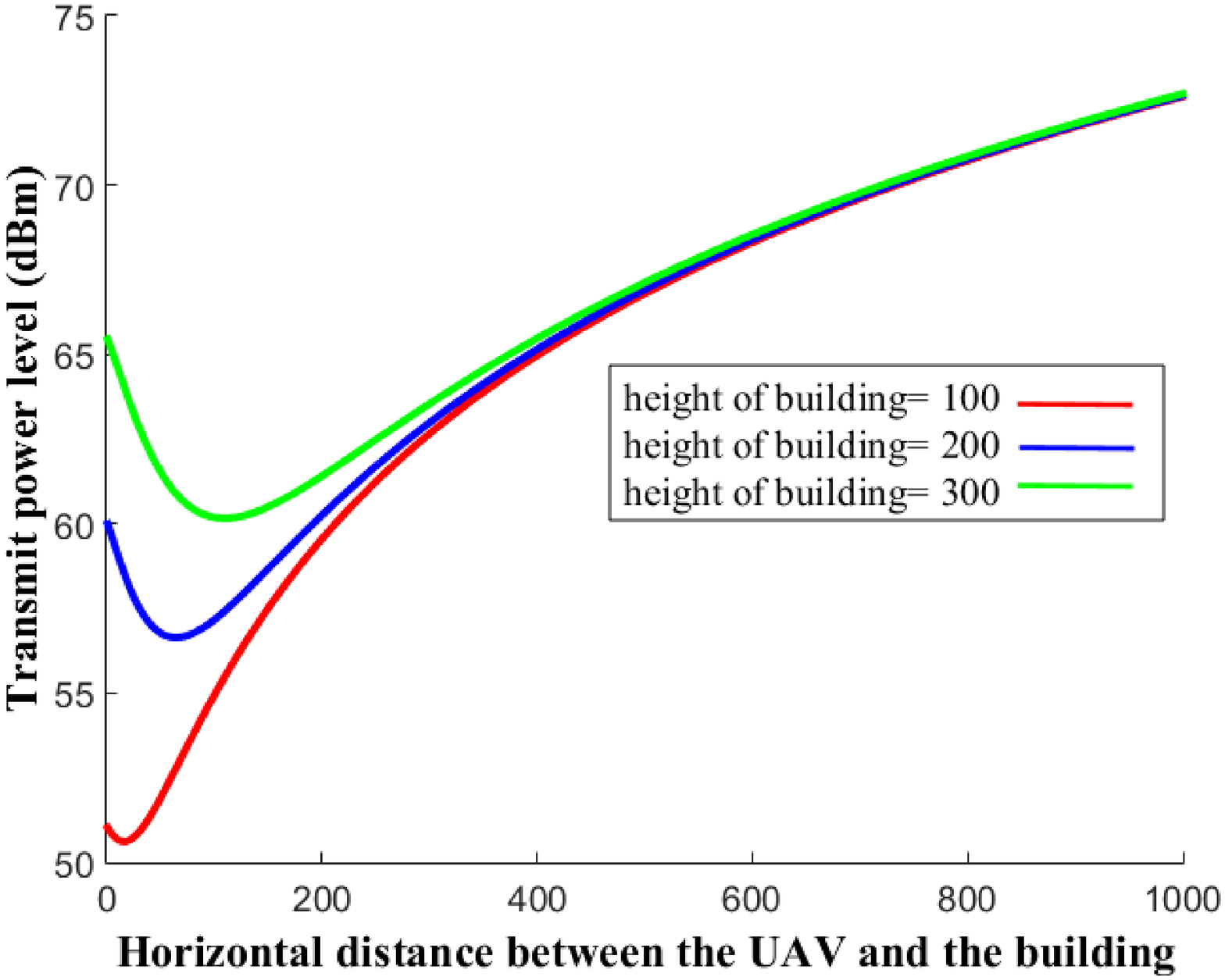}
			\caption{Transmit power required to cover the building 
			}
			\label{fig3}
		\end{minipage}
	\end{figure*}

	\subsection{Problem Formulation }
	Consider a transmission between a UAV located at ($x_{UAV}$,$y_{UAV}$,$z_{UAV}$) and an indoor user $i$ located at ($x_i$,$y_i$,$z_i$). The rate for user $i$ is given by:
	\begin{equation*}
	\begin{split}
	C_{i}=Blog_{2}(1+\dfrac{P_{t,i}/L_i}{N})
	\end{split}
	\end{equation*}
	where $B$ is the transmission bandwidth of the UAV, $P_{t,i}$ is the UAV transmit power to indoor user $i$, $L_i$ is the path loss between a UAV and an indoor user $i$ and $N$ is the noise power. In this paper, we do not explicitly model interference, and instead, implicitly model it as noise.
	
	Let us assume that each indoor user has a channel with bandwidth equals $B/M$, where $M$ is the number of users inside the building and the rate requirement for each user is $v$. Then the minimum power required to satisfy this rate for each user is given by:
	\begin{equation*}
	\begin{split}
	P_{t,i,min}=(2^{\frac{v.M}{B}}-1)\star N\star L_i
	\end{split}
	\end{equation*}
	Our goal is to find the optimal location of UAV such that the total transmit power required to satisfy the downlink rate requirement of each indoor user is minimized. The objective function can be represented as:
	\begin{equation*}
	\begin{split}
	P=\sum_{i=1}^{M}(2^{\frac{v.M}{B}}-1)\star N\star L_i,\\
	\end{split}
	\end{equation*}
	where $P$ is the UAV total transmit power. Since $(2^{\frac{v.M}{B}}-1)\star N$ is constant, our problem can be formulated as:
	\begin{equation*}
	\begin{split}
	\min_{x_{UAV},y_{UAV},z_{UAV}} L_{Total}=\sum_{i=1}^{M}L_i~~~~~~~~~~~~~~~~~~~~~~~~~~~~~~~~~\\
	subject ~to~~~~~~~~~~~~~~~~~~~~~~~~~~~~~~~~~~~~~~~~~~~~~~~~~~~~~~~~~~~\\
	x_{min}\leq x_{UAV}\leq x_{max},~~~~~~~~~~~~~~~~~~~~~~\\
	y_{min}\leq y_{UAV}\leq y_{max},~~~~~~~~~~~~~~~~~~~~~~\\
	z_{min}\leq z_{UAV}\leq z_{max},~~~~~~~~~~~~~~~~~~~~~~\\
	L_{Total}\leq L_{max}~~~~~~~~~~~~~~~~~~~~~~~~~~~
	\end{split}
	\end{equation*}
	Here, the first three constraints represent the minimum and maximum allowed values for $x_{UAV}$, $y_{UAV}$ and  $z_{UAV}$. In the fourth constraint, $L_{max}$ is the maximum allowable path loss and equals $P_{t,max}$$/$$((2^{\frac{v.M}{B}}-1)\star N)$, where $P_{t,max}$ is the maximum transmit power of UAV.
	
	Finding the optimal placement of UAV is generally  difficult because the problem is non-convex. Therefore, in the next subsection, we consider three special cases of practical interest and derive efficient solutions under these cases.
	 \subsection{Efficient Placement of a Single UAV}
	 \label{sec:optimal}
	 \textbf{Case 1.} \textit{The worst location in building:}
	 In this case, we find the minimum transmit power required to cover the building based on the location that has the maximum path loss inside the building. The location that has the maximum path loss in the building is the location that has maximum $d_{out,i}$, maximum $\theta_{i}$, and maximum $d_{in,i}$. The locations that have the maximum path loss are located at the corners of the highest and lowest floors. Since the locations that have the maximum path loss inside the building are the corners of the highest and lowest floors, we place the UAV at the middle of the building ($y_{UAV}$= 0.5$y_b$ and $z_{UAV}$=0.5$z_b$). Then, given Outdoor-to-Indoor path loss models for low-SHF and high-SHF bands, we need to find an efficient horizontal point $x_{UAV}$ for the UAV such that the total transmit power required to cover the building is minimized.
	 \begin{figure*}[ht]
	 	\begin{minipage}[b]{0.315\linewidth}
	 		\centering
	 		\includegraphics[width=\textwidth]{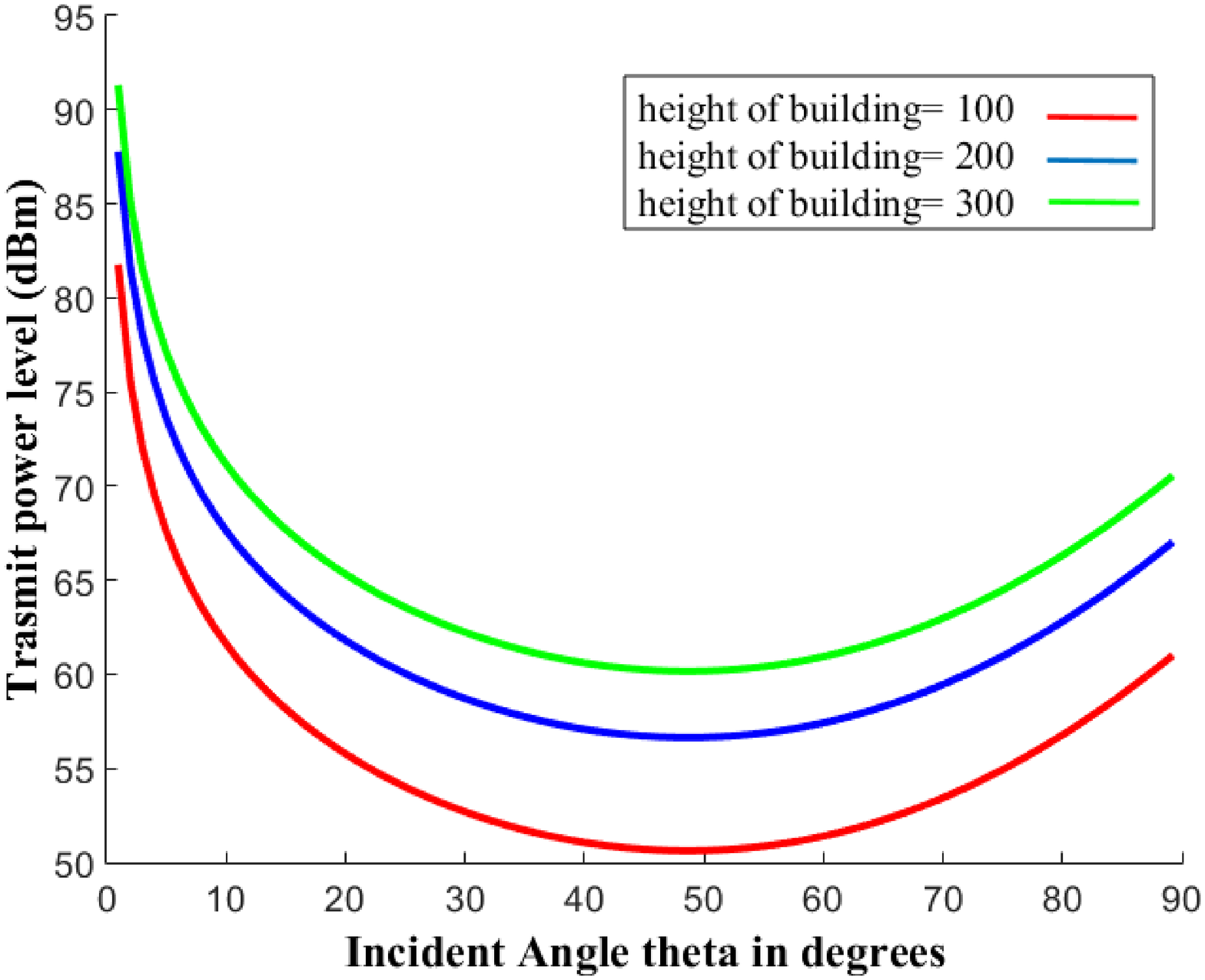}
	 		\caption{Transmit power required to cover the
	 			building, $f_c$=2 GHz
	 		}
	 		\label{fig4}
	 	\end{minipage}
	 	\hspace{0.1cm}
	 	\begin{minipage}[b]{0.315\linewidth}
	 		\centering
	 		\includegraphics[width=\textwidth]{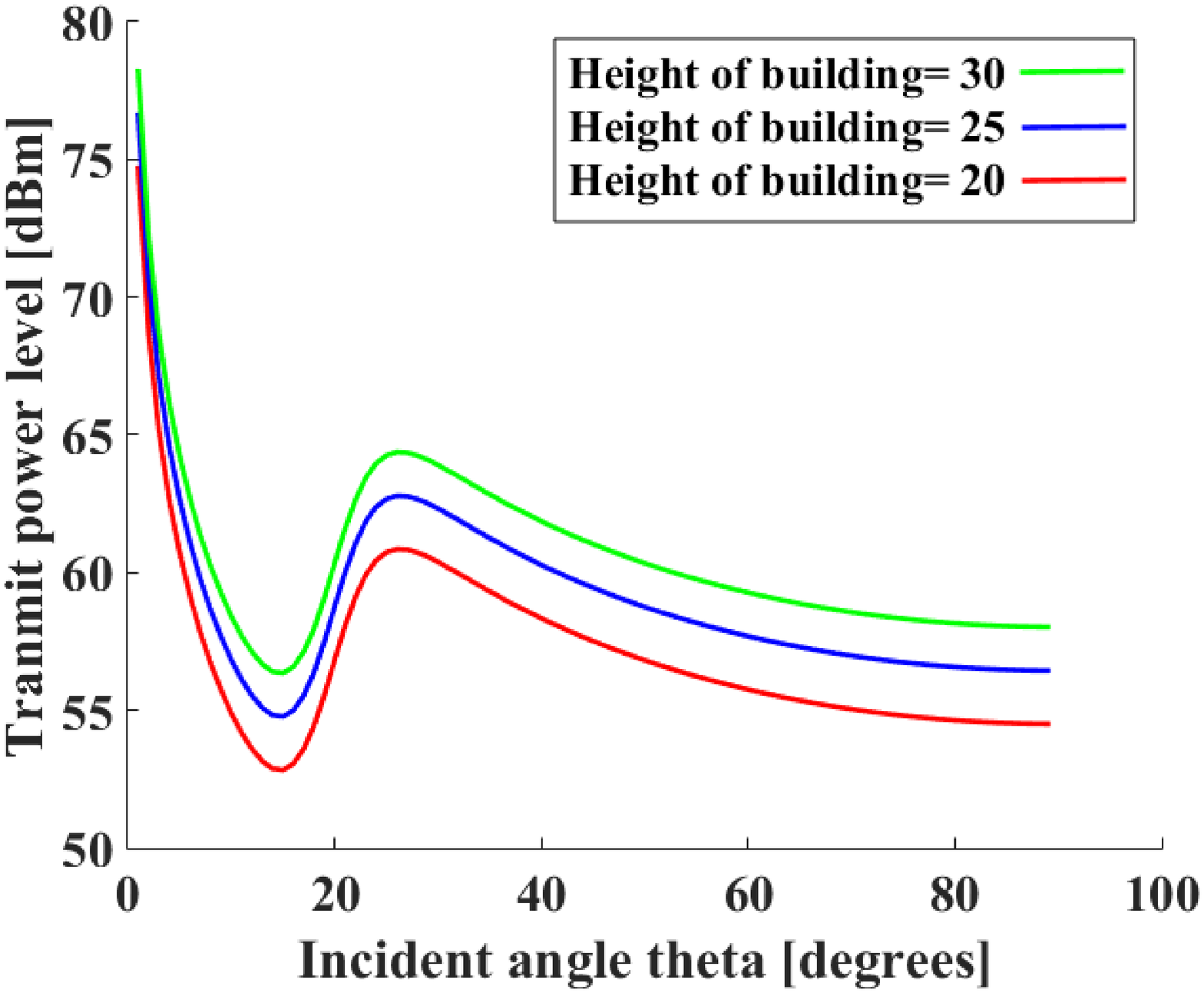}
	 		\caption{Transmit power required to cover the
	 			building, $f_c$=10 GHz
	 		}
	 		\label{fig5}
	 	\end{minipage}
	 	\hspace{0.1cm}
	 	\begin{minipage}[b]{0.315\linewidth}
	 		\centering
	 		\includegraphics[width=\textwidth]{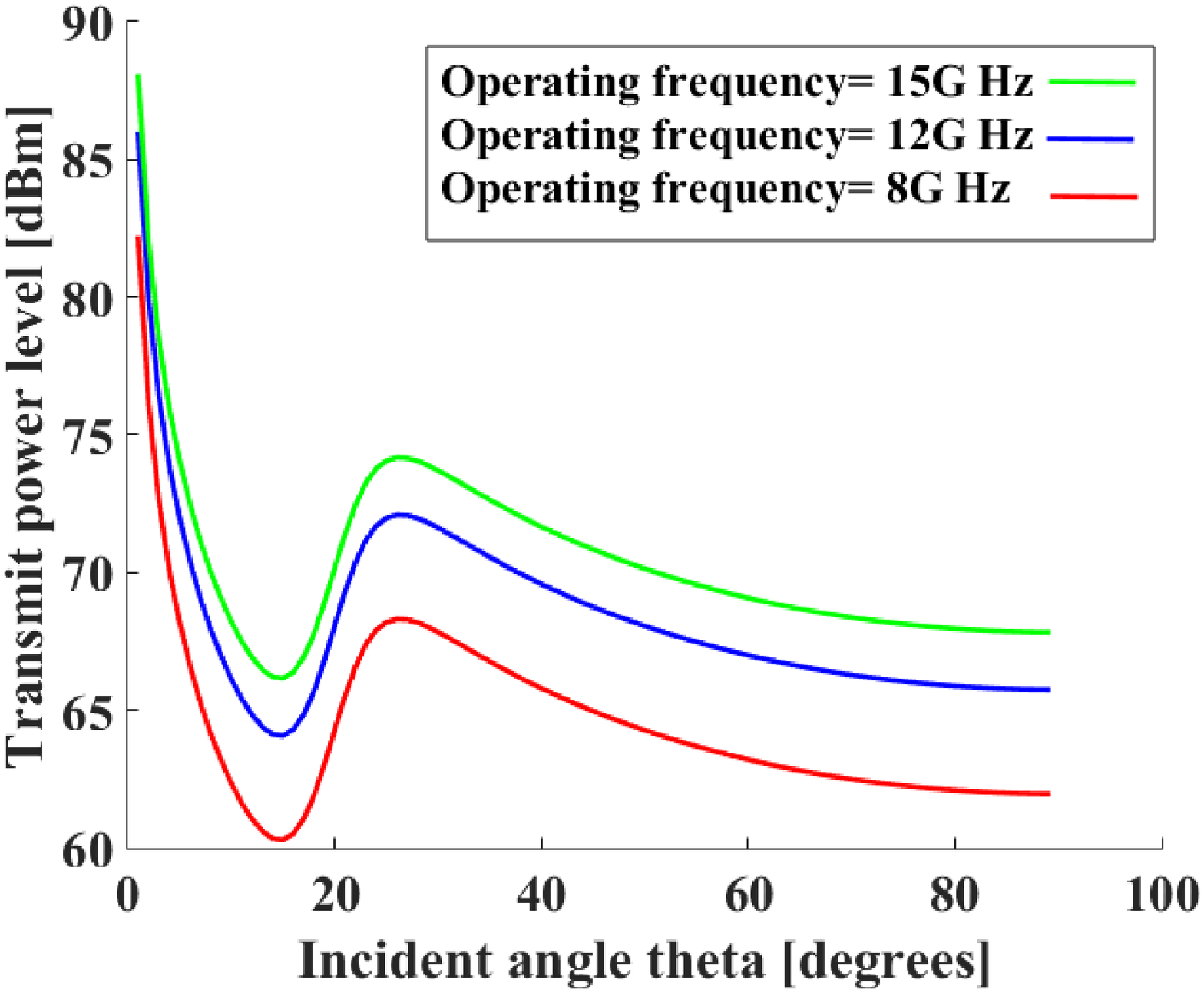}
	 		\caption{Transmit power required to
	 		cover 30 meter building height
	 		}
	 		\label{fig6}
	 	\end{minipage}
	 \end{figure*}
	 
	 Now, when the horizontal distance between the UAV and this location increases, the free space path loss also increases as $d_{out,i}$ increases, while the building penetration loss decreases because we decrease the incident angle $\theta_{i}$. In Figure~\ref{fig3}, we demonstrate the minimum transmit power required to cover a building of different heights, where the minimum transmit power required to cover the building is given by:
	 \begin{equation*}
	 \begin{split}
	 P_{t,min}(dB)=P_{r,th}+L_i
	 \end{split}
	 \end{equation*}
	 \begin{equation*}
	 \begin{split}
	 P_{r,th}(dB)=N+\gamma_{th}
	 \end{split}
	 \end{equation*}
	 Here, $P_{r,th}$ is the minimum received power, $N$ is the noise power (equals -120dBm), $\gamma_{th}$ is the threshold SNR (equals 10dB), $y_b$=50 meters , $x_b$=20 meters and the carrier frequency is 2Ghz. The numerical results show that there is an optimal horizontal point that minimizes the total transmit power required to cover a building. Also, we note that when the height of the building increases, the optimal horizontal distance increases. This is to compensate for the increased building penetration loss due to an increased incident angle.

	 In Theorem 1, we characterize the optimal incident angle $\theta$ for low-SHF band that minimizes the transmit power required to cover the building. This helps us finding the optimal horizontal distance between the UAV and the building. 
	  \begin{theorem}
	  	\label{theorem_one}
	  	For the low-SHF operating frequency case, when we place the UAV at the middle of building , the optimal incident angle $\theta$ that minimizes the transmit power required to cover the building will be equal to $48.654^{o}$ and the optimal horizontal distance between the UAV and the building will be equal to $((\dfrac{0.5z_{b}}{tan(48.654^{o})})^2-(0.5y_{b})^2)^{0.5}-x_{b}$.
	  \end{theorem}
	  \begin{proof}
	  	In order to find the optimal horizontal point, we rewrite the equation that represents the path loss in terms of the incident angle ($\theta_{i}$) and the altitude difference between the UAV and the user $i$ ($\Delta h_{i}$):
	  	\begin{equation*}
	  	\begin{split}
	  	L_{i}(\Delta h_{i},\theta_{i})=w\log_{10}\dfrac{\Delta h_{i}}{\sin\theta_{i}}+w\log_{10}f_{Ghz}+g_{1}\\
	  	+g_{2}+g_{3}(1-\cos\theta_{i})^{2}+g_{4}d_{in,i}
	  	\end{split}
	  	\end{equation*}
	  	We know that the altitude difference between the UAV and the location that has the maximum path loss is constant for a given building. Now, when we take the first derivative with respect to $\theta$ and assign it to zero, we get:
	  	\begin{equation}
	  	\label{ercf1}
	  	\begin{split}
	  	\dfrac{dL(\theta)}{d\theta}=\dfrac{w}{ln10} 
	  	\dfrac{\dfrac{-\Delta h.\cos\theta}{\sin^{2}\theta}}{\dfrac{\Delta h}{\sin\theta}}+2g_{3}\sin\theta(1-\cos\theta)=0 ~~~~~~~~\\
	  	\dfrac{dL(\theta)}{d\theta}=\dfrac{-w}{ln10}\dfrac{\cos\theta}{\sin\theta}+2g_{3}\sin\theta(1-\cos\theta)=0~~~~~~~~~~~~~~~~ \\
	  	\dfrac{w}{ln10}\cos\theta=2g_{3}sin^{2}\theta(1-\cos\theta)~~~~~~~~~~~~~~~~~~~~~~~~~~~~~~~\\
	  	\dfrac{w}{ln10}\cos\theta=2g_{3}(1-\cos^{2}\theta)(1-\cos\theta)~~~~~~~~~~~~~~~~~~~~~~~~\\
	  	2g_{3}\cos^{3}\theta-2g_{3}\cos^{2}\theta-(\dfrac{w}{ln10}+2g_{3})\cos\theta+2g_{3}=0~~~~~~
	  	\end{split}
	  	\end{equation}
	  	To prove that the function is convex, we take the second derivative and we get:
	  	\begin{equation*}
	  	\begin{split}
	  	\dfrac{d^{2}L}{d\theta^{2}}=\dfrac{w}{ln10}\dfrac{1}{\sin^{2}\theta}+2g_{3}\cos\theta(1-\cos\theta)+2g_{3}\sin^{2}\theta>0~~~
	  	for~0<\theta\leq 90
	  	\end{split}
	  	\end{equation*}
	  	 Ecrf~\eqref{ercf1} has only one valid solution which is $\cos\theta$$=$0.6606. Therefore, the optimal incident angle between the UAV and the location that has the maximum path loss inside the building will be $48.654^{o}$.
	  	
	  	In order to find the optimal horizontal distance between the UAV and the building, we apply the pythagorean's theorem. This gives us:
	  	\begin{equation*}
	  	\begin{split}
	  	d_{H}=((\dfrac{0.5z_{b}}{tan(48.654^{o})})^2-(0.5y_{b})^2)^{0.5}
	  	\end{split}
	  	\end{equation*}
	  	Therefore, the optimal horizontal distance between the UAV and the building is given by: 
	  	\begin{equation*}
	  	\begin{split}
	  	d_{opt}=((\dfrac{0.5z_{b}}{tan(48.654^{o})})^2-(0.5y_{b})^2)^{0.5}-x_{b}
	  	\end{split}
	  	\end{equation*}
	  \end{proof}
	  
	  In Figure~\ref{fig4}, we demonstrate the transmit power required to cover the building as a function of the incident angle, we notice that the optimal angle that we characterize in Theorem 1 gives us the minimum transmit power.
	  
	    Now, we find an efficient incident angle $\theta$ for high-SHF band that minimizes the transmit power required to cover the building. In order to find an efficient angle, we rewrite the equation that represents the path loss in terms of the incident angle ($\theta$) and the altitude difference between the UAV and location that has the maximum path loss inside the building ($\Delta h$), we get:
	     \begin{equation*}
	     \begin{split}
	     L(\Delta h,\theta)=(\alpha_1+\alpha_2\log_{10}\frac{\Delta h}{\sin\theta}+\alpha_3\log_{10}f_{Ghz})+~~~\\
	     (\beta_1+\frac{\beta_2-\beta_1}{1+exp(-\beta_3(\theta_i-\beta_4))})+(\gamma_1d_{in,i})
	     \end{split}
	     \end{equation*}
	     By numerically plotting the transmit power required to cover the location that has the maximum path loss inside the building (see Figure~\ref{fig5} and Figure~\ref{fig6}), where $y_b$=50 meters  and $x_b$=20 meters, we show that for different building heights and different operating frequencies there exists only one global minimum value.
	      As can be seen from the figures, to provide wireless coverage to small buildings, the UAV transmit power must be very high, due to the high free space path loss, this demonstrates the need for multiple UAVs to cover the high rise building when we use high-SHF operating frequency. To find an efficient incident angle that could give us the global minimum value, we use the ternary search algorithm. A ternary search algorithm is a method for finding the minimum of a unimodal function, it iteratively splits the domain into three separate regions and discards the one where the minimum does not belong to. The pseudo code of this algorithm is shown in Algorithm 1. From our numerical results, we found that the angle that minimizes the power is always $15^{o}$. This is because the building penetration loss will be minimized at this angle (see Figure~\ref{fig2}). The angles less than $15^{o}$ will also give us minimum building penetration loss but the free space path loss will increase as the incident angle $\theta_i$ decreases. Note that for the high-SHF case the incident angle that results in the minimum path loss is smaller than that for low-SHF case. This is due to the fact that the building penetration loss at high operating frequency will be higher than that at low operating frequency.
	      
	     \begin{algorithm}
	     	\label{alg1}
	     	\begin{algorithmic}
	     		\STATE \textbf{Input:}
	     		\STATE The interval [$a$,$b$] of unimodal function that contains the efficient incident angle. 
	     		\STATE The absolute precision $=$$\mu$.
	     		\STATE \textbf{If} $|$$b$-$a$$|$ $<$ $\mu$:
	     		\STATE ~~~~~\textbf{Return} $\frac{(a+b)}{2}$
	     		\STATE $l$ $=$ $a$$+$$\frac{(b-a)}{3}$
	     		\STATE $r$ $=$ $b$$-$$\frac{(b-a)}{3}$
	     		\STATE \textbf{If} $f(l)$ $>$ $f(r)$
	     		\STATE ~~~~~\textbf{Return} ternary$\_$search($f$, $l$, $b$, $\mu$) 
	     		\STATE \textbf{Else}
	     		\STATE ~~~~~\textbf{Return} ternary$\_$search($f$, $a$, $r$, $\mu$) 
	     	\end{algorithmic}
	     	\caption{Ternary search algorithm}
	     \end{algorithm}
\textbf{Case 2.} \textit{The locations of indoor users are symmetric across the $xy$ and $xz$ planes:}
In this case, we assume that the locations of indoor users are symmetric across the $xy$-plane  ((0,0,0.5$z_b$),($x_b$,0,0.5$z_b$) ,($x_b$,$y_b$,0.5$z_b$),(0,$y_b$,0.5$z_b$))) and the $xz$-plane ((0,0.5$y_b$,0), ($x_b$,0.5$y_b$,0), ($x_b$,0.5$y_b$,$z_b$),(0,0.5$y_b$,$z_b$)). First, we prove that $z_{UAV}$=$0.5z_{b}$ and $y_{UAV}$=$0.5y_{b}$ when the locations of indoor users are symmetric across the $xy$ and $xz$ planes and the operating frequency is low-SHF (Theorem 2) or high-SHF (Theorem 3). Then we use the gradient descent algorithm to find an efficient $x_{UAV}$ that minimizes the transmit power required to cover the building.

   \begin{theorem}
   	For the low-SHF operating frequency case, when the locations of indoor users are symmetric across the $xy$ and $xz$ planes, the optimal ($y_{UAV}$,$z_{UAV}$) that minimizes the power required to cover the indoor users will be equal ($0.5y_{b}$,$0.5z_{b}$).
   \end{theorem}
   	The proof is presented in Appendix A. The question now is how to find an efficient horizontal point $x_{UAV}$ that minimizes the total transmit power.
 In order to find this point, we use the gradient descent algorithm~\cite{sutton1998reinforcement}:
 \begin{equation*}
 \begin{split}
 x_{UAV,n+1}=x_{UAV,n}-a\dfrac{dL_{Total}}{dx_{UAV,n}}\\
 \end{split}
 \end{equation*}
 Where:\\
 \begin{equation*}
 \begin{split}
 \dfrac{dL_{Total}}{dx_{UAV}}=\sum_{i=1}^{M}\dfrac{w}{ln10}\dfrac{-(x_i-x_{UAV})}{d_{out,i}^2}+
 2g_{3}.(1-\dfrac{((x_i-x_{UAV})^2+(y_i-y_{UAV})^2)^{0.5}}{d_{out,i}}).~~~~~~~~~~~~~~~~~\\
 (\dfrac{(x_i-x_{UAV})d_{out,i}((x_i-x_{UAV})^2+(y_i-y_{UAV})^2)^{-0.5}}{{d_{out,i}^2}}-~~~~~~~~~~~~~~~~~~~~~~~~~~~~~~~~~~~\\
 \dfrac{((x_i-x_{UAV})^2+(y_i-y_{UAV})^2)^{0.5}(x_i-x_{UAV})d_{out,i}^{-1}}{d_{out,i}^2})~~~~~~~~~~~~~~~~~~~~~~~~~~~~~~~~~~~~~~~
 \end{split}
 \end{equation*}
 $a$: the step size.\\
 $d_{out,i}$=$((x_i-x_{UAV})^2+(y_i-y_{UAV})^2+(z_i-z_{UAV})^2)^{0.5}$
 
 The pseudo code of this algorithm is shown in Algorithm 2. Now, we prove that $z_{UAV} =0.5z_b$ and $y_{UAV} =0.5y_b$ when the locations of indoor users are symmetric across the xy and xz planes and the operating frequency is high-SHF. 
 
 \begin{algorithm}
 	\label{alg2}
 	\begin{algorithmic}
 		\STATE \textbf{Input:}
 		\STATE The 3D locations of the users inside the building.
 		\STATE The step size $a$, the step tolerance $\epsilon$.
 		\STATE The dimensions of the building  $[0,x_b]$ $\times$ $[0,y_b]$ $\times$ $[0,z_b]$.
 		\STATE The maximum number of iterations $N_{max}$.
 		\STATE \textbf{Initialize} $x_{UAV}$
 		\STATE \textbf{For} $n$=1,2,..., $N_{max}$
 		\STATE  ~~~~~$x_{UAV,n+1}$ $\leftarrow$ $x_{UAV,n}$$-$ $a\dfrac{dL_{Total}}{dx_{UAV,n}}$
 		\STATE ~~~~~~~~~~~~~\textbf{If} $\lVert$ $x_{UAV,n}$ $-$ $x_{UAV,n+1}$ $\rVert$ $<$ $\epsilon$
 		\STATE ~~~~~\textbf{Return:} \textbf{$x_{UAV,opt}$} $=$ $x_{UAV,n+1}$
 		\STATE \textbf{End for}
 	\end{algorithmic}
 	\caption{Efficient $x_{UAV}$ using gradient descent algorithm}
 \end{algorithm}

  \begin{theorem}
  	For the high-SHF operating frequency case, when the locations of indoor users are symmetric across the $xy$ and $xz$ planes, the optimal ($y_{UAV}$, $z_{UAV}$) that minimizes the power required to cover the indoor users will be equal ($0.5y_{b}$,$0.5z_{b}$).
  \end{theorem}
 
  	 The proof is presented in Appendix B. To find an efficient horizontal point $x_{UAV}$ that minimizes the total transmit power, we use the gradient descent algorithm, where:\\
  \begin{equation*}
  \begin{split}
  \dfrac{dL_{Total}}{dx_{UAV}}=\sum_{i=1}^{M}\dfrac{\alpha_2}{ln10}\dfrac{(x_{UAV}-x_i)}{d_{out,i}^2}+
  (\dfrac{-(\beta_2-\beta_1)(\frac{-\beta_3}{\sqrt{1-u^2}})(\frac{-(z_{UAV}-z_i)(x_{UAV}-x_i)}{d_{out,i}^3})}{(1+exp(-\beta_3(sin^{-1}u-\beta_4)))}.~~~~~~~~~~~~~~~~~~\\
  \dfrac{exp(-\beta_3(sin^{-1}u-\beta_4))}{(1+exp(-\beta_3(sin^{-1}u-\beta_4)))})~~~~~~~~~~~~~~~~~~~~~~~~~~~~~~~~~~~~~~~~~~~~~~~~~~~~~~~~~~~~
  \end{split}
  \end{equation*}
  $d_{out,i}$=$((x_i-x_{UAV})^2+(y_i-y_{UAV})^2+(z_i-z_{UAV})^2)^{0.5}$\\
  $u$=$(\dfrac{(z_{UAV}-z_i)}{((x_{UAV}-x_i)^2+(y_{UAV}-y_i)^2+(z_{UAV}-z_i)^2)^{0.5}})$\\
 
 \textbf{Case 3.} \textit{The locations of
 	indoor users are uniformly distributed in each floor:}
 In this case, we propose the Particle Swarm Optimization (PSO)~\cite{kennedy1995particle} to find an efficient 3D placement of the UAV, when the locations of indoor users are uniformly distributed in each floor.
 
 The particle swarm optimization algorithm starts with (npop) random solutions and iteratively tries to improve the candidate solutions based on the best experience of each candidate (particle(i).best.location) and the best global experience (globalbest.location). In each iteration, the best location for each particle (particle(i).best.location) and the best global location (globalbest.location) are updated and the velocities and locations of the particles are calculated based on them~\cite{kalantari2016number}. The velocity value indicates how much the location can be changed (see ecrf~\eqref{ercf2}). The velocity is given by:
 \begin{equation*}
 \begin{split}
 particle(i).velocity=w*particle(i).velocity+~~~~~~~~~~~~~~~~~~~~~~~~~~\\c_1*rand(varsize)*(particle(i).best.location
 -particle(i).location)\\+c_2*rand(varsize)*
 (globalbest.location-particle(i).location)~~~~~
 \end{split}
 \end{equation*}
 where $w$ is the inertia weight, $c_1$ and $c_2$ are the personal and global learning coefficients, and $rand(varsize)$ are random positive numbers. Also, the location of each particle is updated as:
 \begin{equation}
 \begin{split}  
 \label{ercf2}
 particle(i).location=particle(i).location
 +particle(i).velocity
 \end{split}
 \end{equation}

 The pseudo code of the PSO algorithm is shown in Algorithm 3. Convergence of the candidate solutions has been investigated for PSO~\cite{clerc2002particle}. This analyses has resulted in guidelines for selecting a set of coefficients ($\kappa$,$\phi_1$,$\phi_2$) that are believed to cause convergence to a point and prevent divergence of the swarm’s particles. We selected our parameters according to this analysis (see Table~\ref{tableone} and Algorithm 3).

 \begin{algorithm}
 	\label{alg3}
 	\begin{algorithmic}
 		\STATE \textbf{Input:}
 		\STATE The lower and upper bounds of decision variable (varmin,varmax), Construction coefficients ($\kappa$,$\phi_1$,$\phi_2$), Maximum number of iterations (maxit), Population size (npop)
 		\STATE \textbf{Initialiaztion:}
 		\STATE $\phi$=$\phi_1$+$\phi_1$, $\chi$ = ${2\kappa}/{|2-\phi-(\phi^2-4\phi)^{0.5}|}$
 		\STATE $w$=$\chi$, $c_1$=$\chi$$\phi_1$, $c_2$=$\chi$$\phi_2$, globalbest.cost=inf
 		\STATE \textbf{for} i=1:npop
 		\STATE ~~~~~particle(i).location=unifrnd(varmin, varmax, varsize)
 		\STATE~~~~~particle(i).velocity=zeros(varsize)
 		\STATE~~~~~particle(i).cost=costfunction(particle(i).location)
 		\STATE~~~~~particle(i).best.location=particle(i).location
 		\STATE~~~~~particle(i).best.cost=particle(i).cost
 		\STATE~~~~~\textbf{if} particle(i).best.cost $<$ globalbest.cost
 		\STATE~~~~~~~~globalbest=particle(i).best
 		\STATE~~~~~\textbf{end if}
 		\STATE  \textbf{end}
 		\STATE \textbf{PSO Loop:}
 		\STATE \textbf{for} t=1:maxit
 		\STATE~~~~~\textbf{for }i=1:npop
 		\STATE~~~~~~~~particle(i).velocity=$w$*particle(i).velocity+
 		\STATE ~~~~~~~~$c_1$*rand(varsize)*(particle(i).best.location-
 		particle(i).location)+
 		\STATE ~~~~~~~~$c_2$*rand(varsize)*(globalbest.location-particle(i).location)
 		\STATE~~~~~~~~particle(i).location=particle(i).location+
 		particle(i).velocity
 		\STATE~~~~~~~~particle(i).cost=costfunction(particle(i).location)
 		\STATE~~~~~~~~~~ \textbf{if} particle(i).cost $<$ particle(i).best.cost
 		\STATE~~~~~~~~~~~~~~particle(i).best.location = particle(i).location
 		\STATE~~~~~~~~~~~~~~particle(i).best.cost = particle(i).cost
 		\STATE~~~~~~~~~~~~~~~~~~\textbf{if} particle(i).best.cost $<$ globalbest.cost
 		\STATE~~~~~~~~~~~~~~~~~~~~~~globalbest=particle(i).best
 		\STATE~~~~~~~~~~~~~~~~~~\textbf{end if}
 		\STATE~~~~~~~~~~ \textbf{end if}
 		\STATE ~~~~~\textbf{end}
 		\STATE\textbf{end}
 		
 	\end{algorithmic}
 	\caption{Efficient UAV placement using PSO algorithm}
 \end{algorithm}
 
  \section{PROVIDING WIRELESS COVERAGE USING MULTIPLE UAVs}
  Providing wireless coverage to High-rise building using a single UAV can be impractical, due to the limited transmit power of a UAV. The transmit power required to cover the building is too high. It is in the range of 50dBm to 65dBm (see Figures 3, 5 and 6), which corresponds to 100-3000 watts. 
  
 Our problem can be formulated as:
  \begin{equation}
  \begin{array}{ll}
  \displaystyle \min |k| \\
  \textrm{subject to}\\
  \displaystyle \sum_{j=1}^{|k|} y_{ij}=1~~~~~~~~~~~~~~~~~~~~~~~~~~~~~~~~~~~ \forall i \in m &(3.a) \\
  \displaystyle \sum_{i=1}^{|m|}(2^{\frac{v.|m|}{B}}-1).N.L_{ij}.y_{ij} \leq P~~~~~~~~~~~\forall j \in k &(3.b) \\
  x_{min}\leq x_{j}\leq x_{max}~~~~~~~~~~~~~~~~~~~~~~~~~\forall j \in k&(3.c)\\
  y_{min}\leq y_{j}\leq y_{max}~~~~~~~~~~~~~~~~~~~~~~~~~~\forall j \in k&(3.d)\\
  z_{min}\leq z_{j}\leq z_{max}~~~~~~~~~~~~~~~~~~~~~~~~~~\forall j \in k&(3.e)\\
  \end{array}
  \end{equation}
  where $k$ is a set of fully charged UAVs, $m$ is a set of indoor users, $\upsilon$ is the rate requirement for each user (constant), $N$ is the noise power (constant), $B$ is the transmission bandwidth (constant), $L_{ij}$ is the total path loss between UAV $j$ and user $i$ and $P$ is the maximum transmit power of UAV (constant). We also introduce the binary variable $y_{ij}$ that takes the value of 1 if the indoor user $i$  is connected to the UAV $j$ and equals 0 otherwise. The objective is to minimize the number of UAVs that are needed to provide a wireless coverage for indoor users. Constraint set (3.a) ensure that each indoor user should be connected to one UAV. Constraint set (3.b) ensure that the total power consumed by a UAV should not exceed its maximum power consumption limit. Constraints (3.c-3.e) represent the minimum and maximum allowed values for $x_j$, $y_j$ and $z_j$.
  
  \begin{theorem}
  	The problem represented by (3) is NP-complete.
  \end{theorem}
  \begin{proof}
  	The number of constraints is polynomial in terms of
  	the number of indoor users, UAVs and 3D locations. Given any solution for our problem, we can check the solution’s feasibility in polynomial time, then the problem is NP.\\
  	To prove that the problem is NP-hard, we reduce the Bin Packing Problem which is NP-hard~\cite{korf2002new} to a special case of	our problem. In the Bin Packing Problem, we have a set of items $G=\{1,2,..,N\}$ in which each item has volume $z_n$ where $n\in G$. All items must be packed into a finite number of bins ($b_1$, $b$,...,$b_B$), each of volume $V$ in a way that minimizes the number of bins used. The reduction steps are: 1) The $b$-th bin in the Bin Packing Problem is mapped to the $j$-th UAV in our problem, where the volume $V$ for each bin is mapped to the maximum transmit power of the UAV $P$. 2) The $n$-th item is mapped to the indoor $i$-th user, where the volume for each item $n$ is mapped to the power required to cover the $i$-th indoor user. 3) All UAVs have the same maximum transmit power $P$. 4) The power required to cover the $i$-th indoor user from any 3D location will be constant. If there exists a solution to the bin packing problem with cost $C$, then the selected bins will represent the UAVs that are selected and the items in each bin will represent the indoor users that the UAV must cover and the total cost of our problem is $C$.
  \end{proof}
  Due to the intractability of the problem, we study clustering indoor users. In the $k$-means clustering algorithm~\cite{ng2000cs229}, we are given a set of points $m$, and want to group the points into a $k$ clusters such that each point belongs to the cluster with the nearest mean. The first step in the algorithm is to choose the number of clusters $k$. Then, randomly initialize $k$ clusters centroids. In each iteration, the algorithm will do two things:1) Cluster assignment step. 2) Move centroids step. In cluster assignment step, the algorithm goes through each point and chooses the closest centroids and assigns the point to it. In move centroids step, the algorithm calculates the average for each group and moves the centroids there. The algorithm will repeat these two steps until it converges. The algorithm will converge when the assignments no longer change. To find the minimum number of UAVs required to cover the indoor users, we utilize this algorithm to cluster the indoor users. In our algorithm, we assume that each cluster will be covered by only one UAV. We start the algorithm with $k=2$ and after it finishes clustering the indoor users, it applies the particle swarm optimization~\cite{kennedy1995particle} to find the UAV 3D location and UAV transmit power needed to cover each cluster. Then, it checks if the maximum transmit power is sufficient to cover each cluster, if not, the number of clusters $k$ is incremented by one and the  problem is solved again. The pseudo code of this algorithm is shown in Algorithm 4.
  
  \begin{figure*}[!t]
  	\centering
  	\includegraphics[scale=0.115]{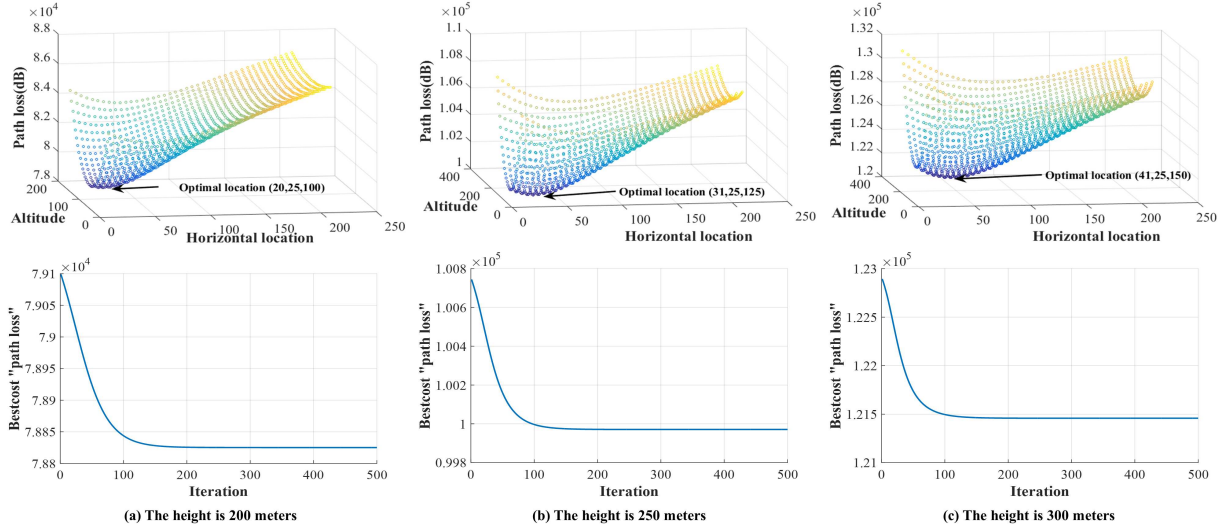}
  	\caption{UAV optimal placement (upper part)  and convergence speed of the GD algorithm (lower part) for different building heights, $f_c=2G~Hz$}
  	\label{fig7}
  \end{figure*}
  
  \begin{figure*}[!t]
  	\centering
  	\includegraphics[scale=0.115]{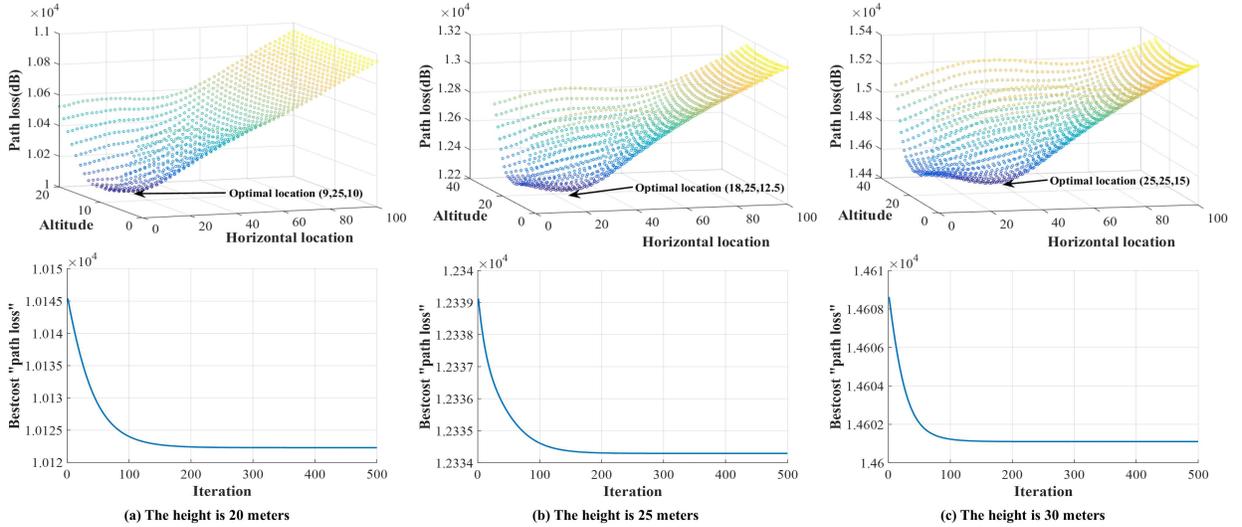}
  	\caption{UAV optimal placement (upper part)  and convergence speed of the GD algorithm (lower part) for different building heights, $f_c=15G~Hz$}
  	\label{fig8}
  \end{figure*}
   \begin{algorithm}
   	\label{alg4}
   	\begin{algorithmic}
   		\STATE \textbf{Input:}
   		\STATE The maximum transmit power of UAV $(P)$.
   		\STATE The 3D locations of $m$ indoor users $(x_i,y_i,z_i)$.
   		\STATE Number of clusters $(|k|=2)$. 
   		\STATE \textbf{START:}
   		\STATE 1: Initialize cluster centroids $\gamma_1,\gamma_2,...,\gamma_k\in R^{n}$ randomly.
   		\STATE 2: Repeat until convergence: 
   		\STATE ~~~~~For every indoor user $i\in m$, set \STATE ~~~~~~~~~~$c^{(i)}=arg~\min\limits_{j\in k}||(x_i,y_i,z_i)-\gamma_j||^2$
   		\STATE ~~~~~For each cluster $j\in k$, set
   		\STATE ~~~~~~~~~~$\gamma_j=\frac{\displaystyle\sum_{i\in m,c^{(i)}=j}(x_i,y_i,z_i)}{\displaystyle\sum_{i\in m,c^{(i)}=j}1}$
   		\STATE 3: Using particle swarm optimization algorithm, calculate the UAV efficient 3D location and the transmit power for each cluster $j\in k$:
   		\STATE ~~~~~$P(j)=\displaystyle\sum_{i\in m,c^{(i)}=j}(2^{\frac{v.|m|}{B}}-1)\star N\star L_i$ 
   		\STATE 4: \textbf{For} $j=1~to~|k|$
   		\STATE ~~~~~ \textbf{If} ($P(j)>P$)
   		\STATE ~~~~~~~~~~$|k|=|k|+1$
   		\STATE ~~~~~ $\textbf{go to}~~\textbf{START}$
   		\STATE ~~~~\textbf{End}
   		\STATE \textbf{Output:}
   		\STATE $|k|$ Clusters.
   		\STATE The transmit Power of each UAV.
   		\STATE The 3D locations of UAVs.
   	\end{algorithmic}
   	\caption{Clustering Indoor Users}
\end{algorithm}
\vspace{3mm}
 \section{Numerical Results}
 \label{sec:results}
 \subsection{Simulation results for single UAV }
  First, we verify our results for the second case, when the locations of indoor users are symmetric across the $xy$ and $xz$ planes, using different operating frequencies, 2GHz for low-SHF band and 15GHz for high-SHF. We assume that each floor contains 20 users. Then we apply the gradient descent (GD) algorithm to find the optimal horizontal point $x_{UAV}$ that minimizes the transmit power required to cover the indoor users. Table~\ref{tableone} lists the parameters used in the numerical analysis for single UAV cases.
 \begin{table}[!h]
 	\scriptsize
 	\renewcommand{\arraystretch}{1.3}
 	\caption{Parameters in numerical analysis for single UAV}
 	\label{tableone}
 	\centering
 	\begin{tabular}{|c|c|}
 		\hline
 		Vertical width of building $y_{b}$ & 50 meters\\
 		\hline
 		Hight of each floor & 5 meters\\
 		\hline
 		Step size $a$ "GD algorithm"& 0.01\\
 		\hline
 		Maximum number of iterations $N_{max}$ "GD algorithm"& 500\\
 		\hline 
 		The carrier frequency $f_{Ghz}$, low-SHF & 2Ghz\\
 		\hline 
 		The carrier frequency $f_{Ghz}$, high-SHF & 15Ghz\\
 		\hline 
 		Number of users in each floor & 20 users\\
 		\hline
 		(varmin,varmax) "PSO algorithm"&(0,1000)\\
 		\hline
 		($\kappa$,$\phi_1$,$\phi_2$) "PSO algorithm" & (1,2.05,2.05)\\
 		\hline
 	\end{tabular}
 \end{table}

 \begin{figure*}[!t]
 	\centering
 	\includegraphics[scale=0.115]{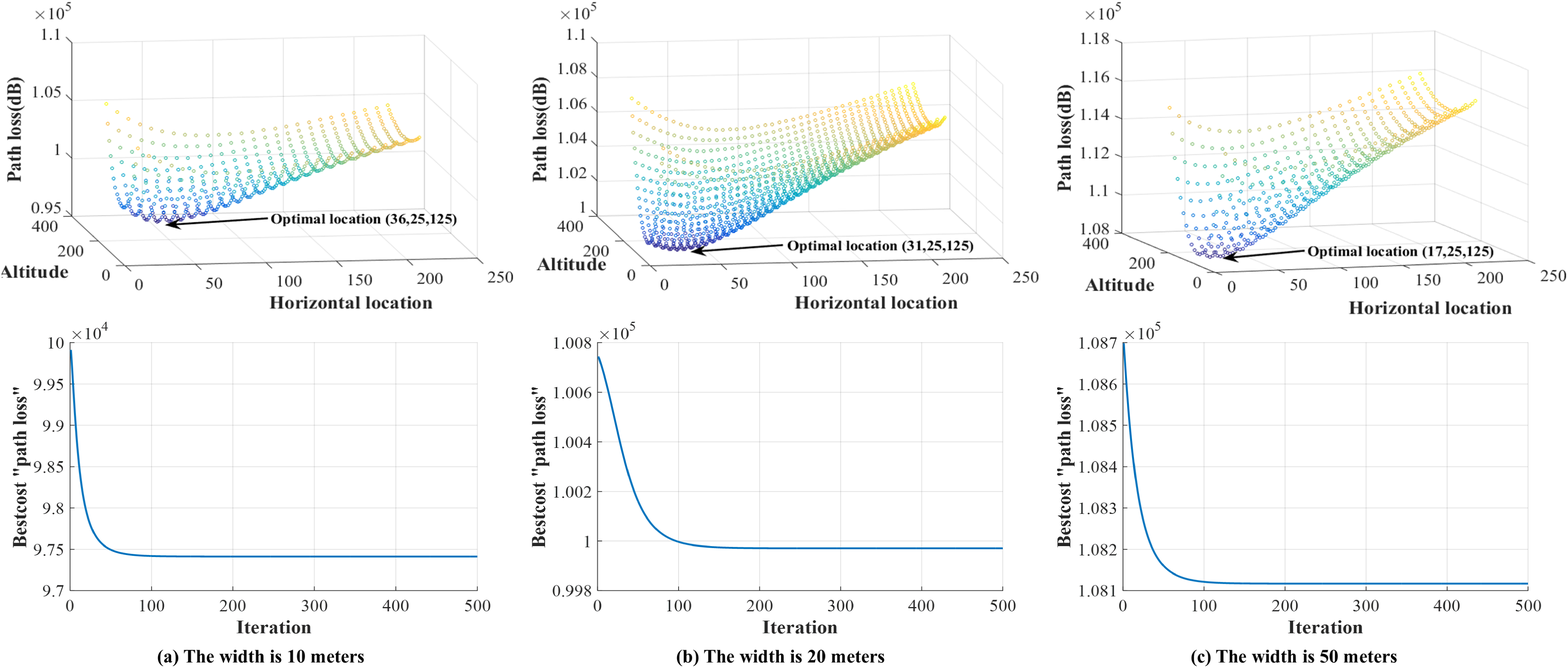}
 	\caption{UAV optimal placement (upper part)  and convergence speed of the GD algorithm (lower part) for different building widths, $f_c=2G~Hz$ }
 	\label{fig9}
 \end{figure*}
 
  \begin{figure*}[!t]
  	\centering
  	\includegraphics[scale=0.115]{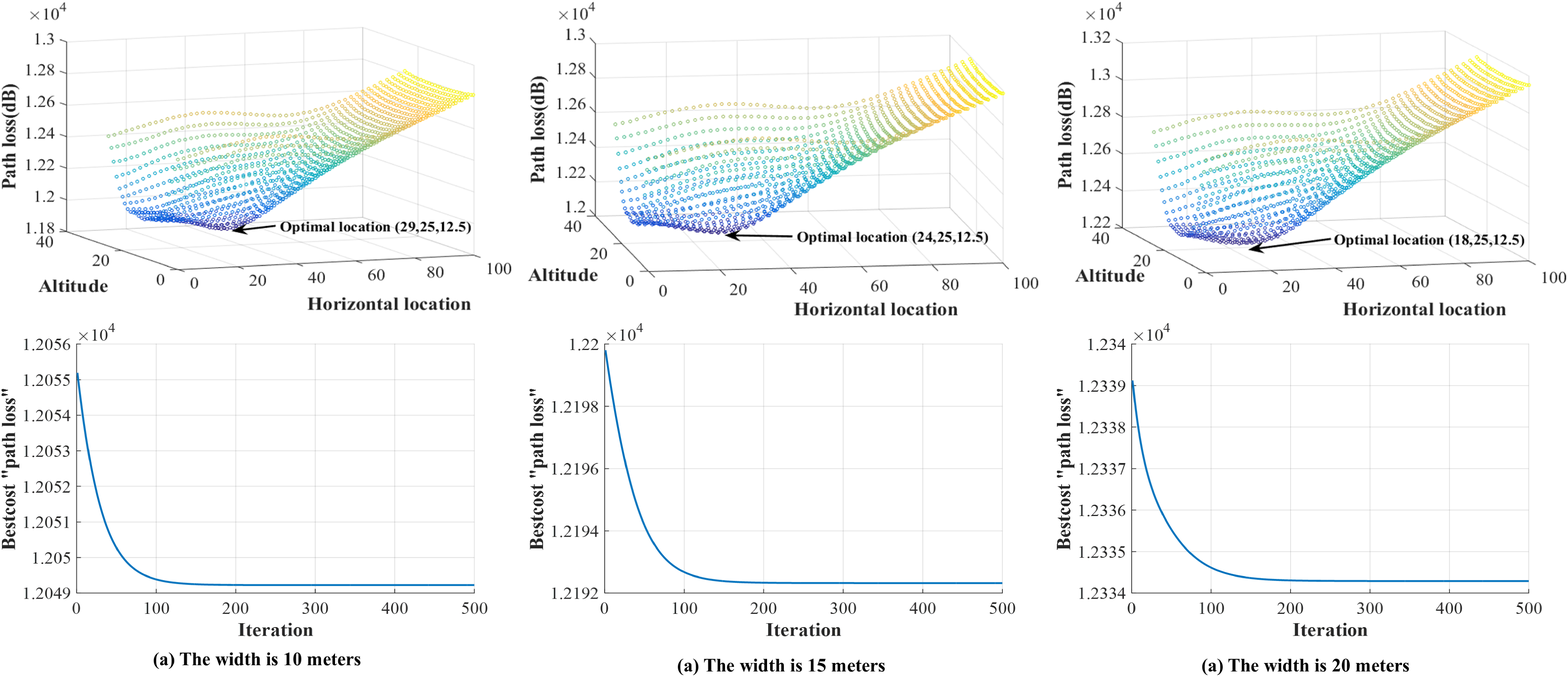}
  	\caption{UAV optimal placement (upper part)  and convergence speed of the GD algorithm (lower part) for different building widths, $f_c=15G~Hz$}
  	\label{fig10}
  \end{figure*}
 \begin{table*}[!t]
 	\scriptsize
 	\renewcommand{\arraystretch}{1.55}
 	\caption{Simulation Results}
 	\label{tabletwo}
 	\centering
 	\begin{tabular}{|c||c|c|c|c|c|c|}
 		\hline \hline
 		Algorithm&Distribution&Building height &Horizontal building&Vertical building &Efficient 3D placement &Efficient total \\
 		&&$z_b$& width $x_b$&width $y_b$&&path loss(dB) \\
 		\hline \hline
 		GD&symmetric&200&20&50&(20.025, 25, 100)&$7.8825*10^{4}$ \\
 		\hline
 		PSO &symmetric&200&20&50&(20.040, 25.0130, 100.0015)&$7.8825*10^{4}$\\
 		\hline \hline
 		GD&symmetric&250&20&50&(30.809, 25, 125)&$9.9971*10^{4}$ \\
 		\hline
 		PSO &symmetric&250&20&50&(30.736 , 24.960, 124.956)&$9.9971*10^{4}$\\
 		\hline \hline
 		GD&symmetric&300&20&50&(40.746, 25, 150)&$1.2146*10^{5}$ \\
 		\hline
 		PSO &symmetric&300&20&50&(40.758, 25.048, 150.054)&$1.2146*10^{5}$\\
 		\hline \hline
 		
 		\hline \hline
 		GD&uniform&200&20&50&(24.725, 25, 100)&$7.8853*10^{4}$ \\
 		\hline
 		PSO &uniform&200&20&50&(21.799, 37.389, 111.790)&$7.8645*10^{4}$\\
 		\hline \hline
 		GD&uniform&250&20&50&(33.818, 25, 125)&$9.9855*10^{4}$ \\
 		\hline
 		PSO &uniform&250&20&50&(32.921, 28.712, 124.029)&$9.9725*10^{4}$\\
 		\hline \hline
 		GD&uniform&300&20&50&(43.117, 25, 150)&$1.2154*10^{5}$ \\
 		\hline
 		PSO &uniform&300&20&50&(46.589, 31.506 ,143.858)&$1.2117*10^{5}$\\
 		\hline \hline

 		\hline \hline
 		GD&uniform&250&10&50&(38.521, 25, 125)&$9.7413*10^{4}$ \\
 		\hline
 		PSO &uniform&250&10&50&(32.104, 21.017, 129.266)&$9.7252*10^{4}$\\
 		\hline \hline
 		GD&uniform&250&30&50&(29.393, 25, 125)&$1.0275*10^{5}$ \\
 		\hline
 		PSO &uniform&250&30&50&(25.529, 4.938, 138.765)&$1.0211*10^{5}$\\
 		\hline \hline
 		GD&uniform&250&50&50&(22.711, 25, 125)&$1.0753*10^{5}$ \\
 		\hline
 		PSO &uniform&250&50&50&(14.548, 17.308 ,131.8940)&$1.0696*10^{5}$\\
 		\hline \hline
 		
 	\end{tabular}
 \end{table*}
 
 \begin{figure*}[!t]
 	\centering
 	\includegraphics[scale=0.115]{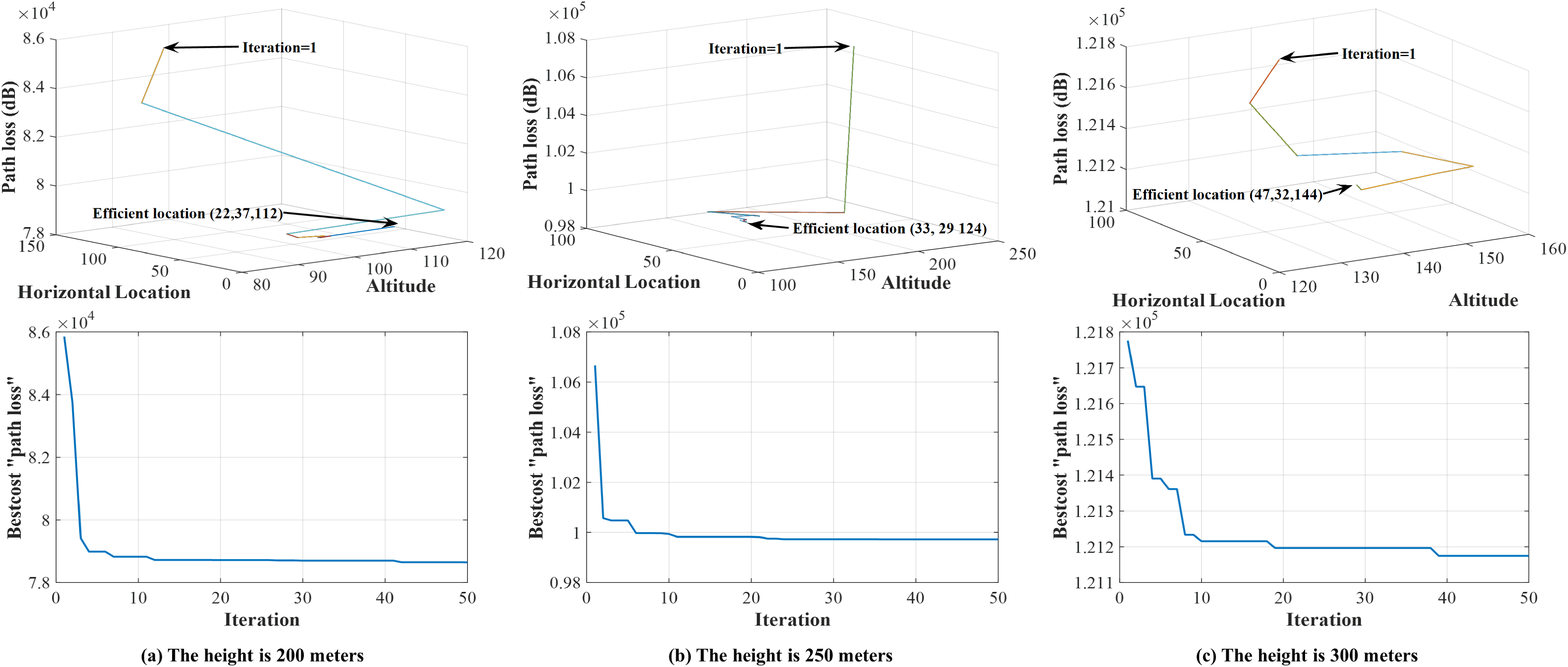}
 	\caption{UAV efficient placement (upper part)  and convergence speed of the PSO algorithm (lower part) for different building heights}
 	\label{fig:fig11}
 \end{figure*}
 
 \begin{figure*}[!t]
 	\centering
 	\includegraphics[scale=0.115]{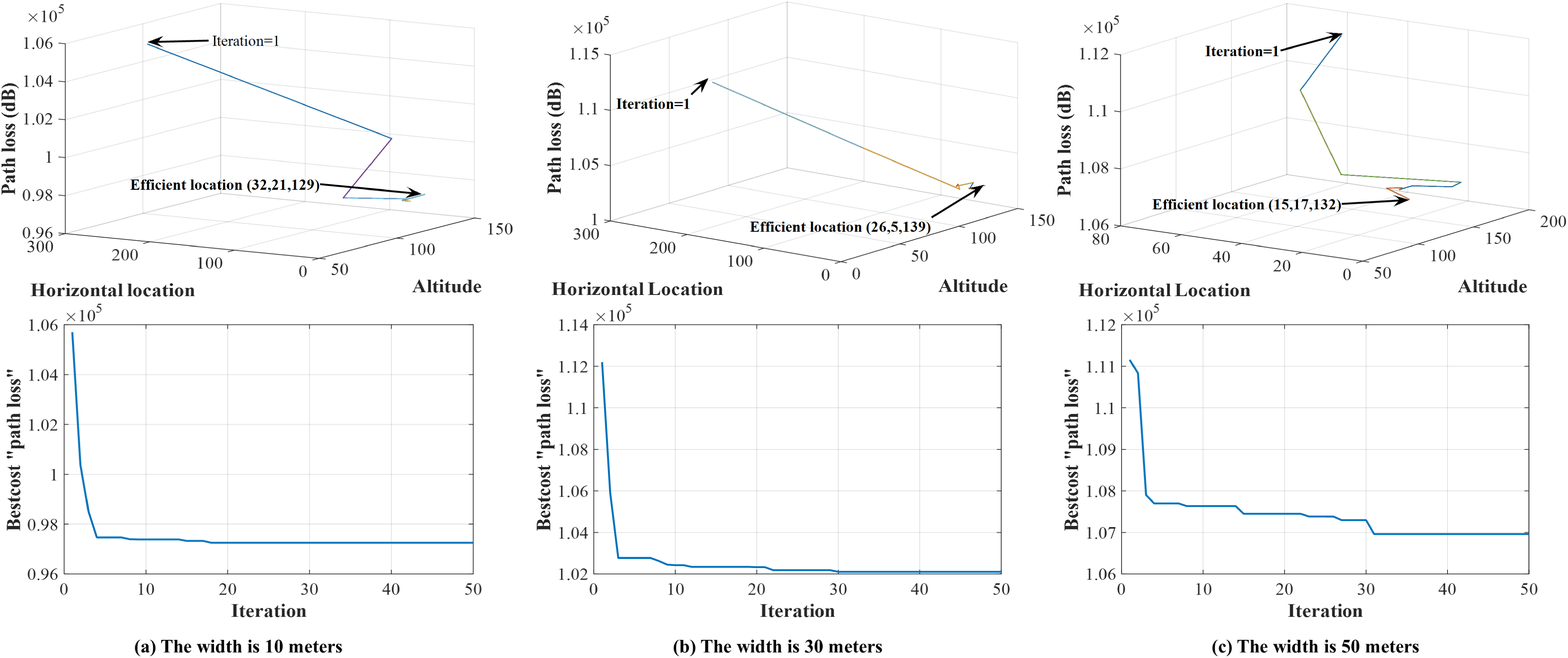}
 	\caption{UAV efficient placement (upper part)  and convergence speed of the PSO algorithm (lower part) for different building widths}
 \label{fig12}	
 \end{figure*}

 In Figures~\ref{fig7} and ~\ref{fig8}, we find the optimal horizontal points for a building of different heights. In the upper part of the figures, we find the total path loss at different locations ($x_{UAV}$,0.5$y_{b}$,$z_{UAV}$) and the optimal horizontal point $x_{UAV}$ that results in the minimum total path loss using the GD algorithm. In the lower part of the figures, we show the convergence speed of the GD algorithm. As can be seen from the figures, when the height of the building increases, the optimal horizontal point $x_{UAV}$ increases. This is to compensate the increased building penetration loss due to an increased incident angle. 
 
 In Figures~\ref{fig9} and ~\ref{fig10}, we investigate the impact of different building widths (i.e., $x_b$). We fix the building height to be 250 meters for low-SHF operating frequency and 25 meters for high-SHF, then we vary the building width. As can be seen from the figures, when the building width increases, the optimal horizontal distance decreases. This is to compensate for the increased indoor path loss due to an increased building width.
 
 Now, we validate the simulation results for low-SHF operating frequency by using the particle swarm optimization (PSO) algorithm and verify our result for the third case, when the locations of indoor users are uniformly distributed in each floor, using low-SHF operating frequency. As can be seen from
 the simulation results in Table II, both algorithms converge to the same 3D placement, when the locations of indoor users are symmetric across the $xy$ and $xz$ planes.
 
  After that, we assume that each floor contains 20 users and the locations of these users are uniformly distributed in each floor. When we apply the GD algorithm, the 3D efficient placements and the total costs for 200 meter, 250 meter and 300 meter buildings are (24.7254, 25, 100) ($7.8853*10^{4}$), (33.8180, 25, 125) ($9.9855*10^{4}$) and (43.1170, 25, 150)($1.2154*10^{5}$), respectively. UAV efficient placement and the convergence speed of the PSO algorithm for different building heights is shown in Figure 11. The 3D efficient placements and the total costs for 200 meter, 250 meter and 300 meter buildings are (21.7995, 37.3891, 111.7901) ($7.8645*10^{4}$), (32.9212, 28.7125, 124.0291) ($9.9725*10^{4}$) and (46.5898, 31.5061 ,143.8588)($1.2117*10^{5}$), respectively. As can be seen from the simulation results, the PSO algorithm provides better results. It provides total cost less than the cost that the GD algorithm provides by (37dB-208dB). This is because the PSO algorithm is designed for the case in which the locations of indoor users are uniformly distributed in each floor. On the other hand, the GD algorithm is designed for the case in which the locations of indoor users are symmetric across the dimensions of each floor.
  
  We also investigate the impact of different building widths (i.e., $x_b$) using the GD and PSO algorithms (see Figure~\ref{fig12}). We fix the building height to be 250 meters and vary the building width.  As can be seen from the simulation results, the PSO algorithm provides better results. It provides total cost less than the cost that the GD algorithm provides by (57dB-161dB).
  
   We can notice that the tradeoff in case three is similar to that in case two, when the height of the building increases, the efficient horizontal point $x_{UAV}$ computed by our algorithm increases. This is to compensate the increased building penetration loss due to an increased incident angle. Also, when the building width increases, the efficient horizontal distance computed by our algorithm decreases. This is to compensate the increased indoor path loss due to an increased building width.

\subsection{Simulation results for multiple UAVs }
In this section, we verify our results for multiple UAVs scenario. First, we assume that a building will host a special event (such as concert, conference, etc.), the dimensions of the building are $[0,20]\times[0,50]\times[0,100]$. The organizers of the event reserve all floors higher than 75 meters and they expect that 200 people will attend the event. Due to interference from near-by macro cells, the organizers decide to use UAVs to provide wireless coverage to the indoor users. We assume that 200 indoor users are uniformly distributed in upper part of the building (higher than 75 meters) and 200 indoor users are uniformly distributed in the lower part (less than 75 meters). Then, we apply the clustering indoor users algorithm to find the minimum number of UAVs required to cover the indoor users. Table III lists the parameters used in the numerical analysis for multiple UAVs.\\

The algorithm starts with $k=2$ and after it finishes clustering the indoor users, it applies the particle swarm optimization to find the UAV 3D location and UAV transmit power needed to cover each cluster. Then, it checks if the maximum transmit power is sufficient to cover each cluster, if not, the number of clusters $k$ is incremented by one and the  problem is solved again. As can be seen from the simulation results in Figure~\ref{fig13}, we need 5 UAVs to cover the indoor users. We can notice that an efficient horizontal point $x_{UAV}$ for all UAVs 3D locations is the same $x_{UAV}=25$, the minimum allowed value for $x_{UAV}$, this is because the tradeoff (shown in Figure~\ref{fig3}) disappears when a UAV covers small height of building.

In Figure~\ref{fig14}, we uniformly split the building into $k$ parts and cover it by $k$ UAVs. As can be seen from the simulation results, we need 9 UAVs to cover the indoor users. The clustering algorithm provides better results, this is because it utilizes the distribution of indoor users to divide them into clusters. On the other hand, the uniformly split method is designed for the case in which the locations of indoor users uniformly distributed in the building.

\begin{table}[!h]
	\scriptsize
	\renewcommand{\arraystretch}{1.3}
	\caption{Parameters in numerical analysis for multiple UAVs}
	\label{tablethree}
	\centering
	\begin{tabular}{|c|c|}
		\hline
		 Maximum transmit power of UAV ($P$) & 5 Watt\\
		\hline
	Operating frequency ($f$)& 2Ghz\\
		\hline
  	 Transmission bandwidth	($B$) & 50M Hz\\
		\hline
	Rate requirement for each user ($\upsilon$) & 2.2Mbps\\
		\hline
		Noise power ($N$) & -150 dBm\\
		\hline
	Min and Max allowed values for $x_j$,$[x_{min},x_{max}]$ & [25,1000]\\
		\hline
		Min and Max allowed values for $y_j$, $[y_{min},y_{max}]$& [0,50] \\
		\hline
		Min and Max allowed values for $z_j$, $[z_{min},z_{max}]$ & [0,1000] \\
		\hline 
	\end{tabular}
\end{table}

\begin{figure}[!h]
	\centering
	\includegraphics[scale=0.5]{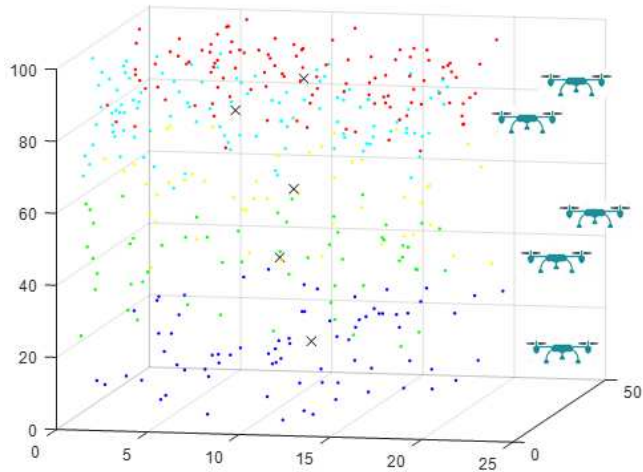}
	\caption{UAVs efficient placements using clustering algorithm}
	\label{fig13}				
\end{figure}

\begin{figure}[!h]
	\centering
	\includegraphics[scale=0.5]{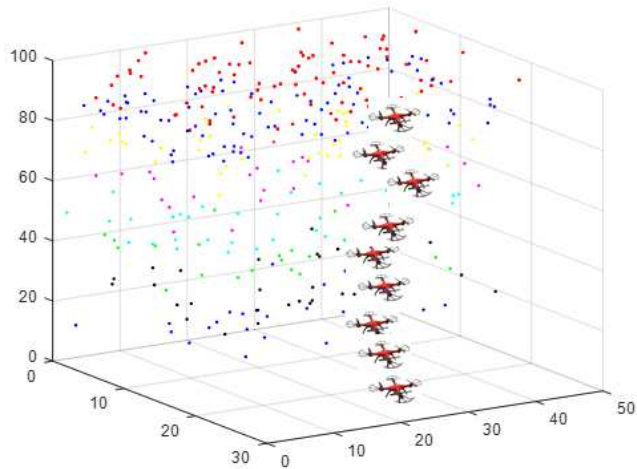}
	\caption{UAVs efficient placements using uniform split method}
	\label{fig14}				
\end{figure}

\section{Conclusion}
In this paper, we study the problem of providing wireless
coverage for users inside a high-rise building using UAVs. First, we demonstrate why the Air-to-Ground path loss model is not appropriate for considering indoor users with
3D locations. Then, we present Outdoor-to-Indoor path
loss models, show the tradeoff in these models, and study the
problem of minimizing the transmit power required to cover
the building. Due to the intractability of the problem, we study
an efficient placement of a single UAV under three cases. Due to the limited transmit power of a UAV, we formulate the problem of minimizing the number of UAVs required to provide wireless coverage to high rise building and prove that this problem is NP-complete. Due to the intractability of the problem, we use clustering to minimize the number of UAVs required to cover the indoor users. In order to model more realistic scenarios, we will study the problem of providing wireless coverage for multiple buildings in our future work.

\appendices
\section{Proof of Theorem 2}
Consider that $m_{1}$ represents the users that have altitude lower than the UAV altitude and $m_{2}$ represents the users that have altitude higher than the UAV altitude, then:
\begin{equation*}
\begin{split}
d_{out,i}=((x_{UAV}-x_i)^2+(y_{UAV}-y_i)^2+(z_{UAV}-z_i)^2)^{0.5},~~\forall z_{UAV}>z_i\\  
d_{out,i}=((x_{UAV}-x_i)^2+(y_{UAV}-y_i)^2+(z_i-z_{UAV})^2)^{0.5},~~\forall z_{UAV}<z_i\\
\end{split}
\end{equation*}
Also, 
\begin{equation*}
\begin{split}
cos_{\theta_i}=\dfrac{((x_{UAV}-x_i)^2+(y_{UAV}-y_i)^2)^{0.5}}{((x_{UAV}-x_i)^2+(y_{UAV}-y_i)^2+(z_{UAV}-z_i)^2)^{0.5}},~~\forall z_{UAV}>z_i\\ 
cos_{\theta_i}=\dfrac{((x_{UAV}-x_i)^2+(y_{UAV}-y_i)^2)^{0.5}}{((x_{UAV}-x_i)^2+(y_{UAV}-y_i)^2+(z_i-z_{UAV})^2)^{0.5}},~~\forall z_{UAV}<z_i\\
\end{split}
\end{equation*}
Rewrite the total path loss:
\begin{equation*}
\begin{split}
L_{Total}=
\sum_{i=1}^{m_1}(wlog_{10}(d_{out,i})+g_{3}(1-\cos\theta_{i})^{2})+
\sum_{i=1}^{m_2}(wlog_{10}(d_{out,i})+g_{3}(1-\cos\theta_{i})^{2})+K
\end{split}
\end{equation*}
Where:
\begin{equation*}
\begin{split}
K=\sum_{i=1}^{M}(wlog_{10}f_{Ghz}+g_{1}+g_{2}+g_{4}d_{in,i})~~~~~~~~~~~~~~~~~~~~~~~~~~~~~~~~~~~~~~~~~~~~~~~~~~~~~~~~
\end{split}
\end{equation*}
Now, take the derivative with respect to $z_{UAV}$, we get:
\begin{equation*}
\begin{split}
\dfrac{dL_{Total}}{dz_{UAV}}=
\sum_{i=1}^{m_1}\dfrac{w}{ln10}\frac{(z_{UAV}-z_i)}{((x_{UAV}-x_i)^2+(y_{UAV}-y_i)^2+(z_{UAV}-z_i)^2)}
+~~~~~~~~~~~~~~~~~~~~~~~~~\\2g_{3}.
(1-\dfrac{((x_{UAV}-x_i)^2+(y_{UAV}-y_i)^2)^{0.5}}{((x_{UAV}-x_i)^2+(y_{UAV}-y_i)^2+(z_{UAV}-z_i)^2)^{0.5}}).~~~~~~~~~~~~~~~~~~~~~~~~\\
(\dfrac{((x_{UAV}-x_i)^2+(y_{UAV}-y_i)^2)^{0.5}(z_{UAV}-z_i)}{((x_{UAV}-x_i)^2+(y_{UAV}-y_i)^2+(z_{UAV}-z_i)^2)^{\frac{3}{2}}})+~~~~~~~~~~~~~~~~~~~~~~~~~~~~~~~~\\ 
\sum_{i=1}^{m_2}\dfrac{w}{ln10}\frac{-(z_{i}-z_{UAV})}{((x_{UAV}-x_i)^2+(y_{UAV}-y_i)^2+(z_i-z_{UAV})^2)}~~~~~~~~~~~~~~~~~~~~~~~~~~~~\\
+2g_{3}.
(1-\dfrac{((x_{UAV}-x_i)^2+(y_{UAV}-y_i)^2)^{0.5}}{((x_{UAV}-x_i)^2+(y_{UAV}-y_i)^2+(z_i-z_{UAV})^2)^{0.5}}).~~~~~~~~~~~~~~~~~~~~~\\
(\dfrac{-((x_{UAV}-x_i)^2+(y_{UAV}-y_i)^2)^{0.5}(z_i-z_{UAV})}{((x_{UAV}-x_i)^2+(y_{UAV}-y_i)^2+(z_i-z_{UAV})^2)^{\frac{3}{2}}})~~~~~~~~~~~~~~~~~~~~~~~~~~~~~~~~~~
\end{split}
\end{equation*}

Rewrite the $\dfrac{dL_{Total}}{dz_{UAV}}$ again, we have:

\begin{equation*}
\begin{split}
\dfrac{dL_{Total}}{dz_{UAV}} =\sum_{i=1}^{m_1}\dfrac{w}{ln10}\frac{(z_{UAV}-z_i)}{d_{out,i}^{2}}+
2g_{3}.(1-\dfrac{((x_{UAV}-x_i)^2+(y_{UAV}-y_i)^2)^{0.5}}{d_{out,i}}).~~~~~~~~~~~~~~~~~~~~\\
(\dfrac{((x_{UAV}-x_i)^2+(y_{UAV}-y_i)^2)^{0.5}(z_{UAV}-z_i)}{d_{out,i}^{3}})+
\sum_{i=1}^{m_2}\dfrac{w}{ln10}\frac{-(z_i-z_{UAV})}{d_{out,i}^{2}}+~~~~~~~~~~~~~~~~~~~~~~~\\
2g_{3}.(1-\dfrac{((x_{UAV}-x_i)^2+(y_{UAV}-y_i)^2)^{0.5}}{d_{out,i}}).
(\dfrac{-((x_{UAV}-x_i)^2+(y_{UAV}-y_i)^2)^{0.5}(z{i}-z_{UAV})}{d_{out,i}^{3}})~~~\\ 
\end{split}
\end{equation*}
The equation above equals zero when the UAV altitude equals the half of the building height, where the locations of indoor users are symmetric across the $xy$ and $xz$ planes.

\section{Proof of Theorem 3}
Consider that $m_{1}$ represents the users that have altitude lower than the UAV altitude and $m_{2}$ represents the users that have altitude higher than the UAV altitude, then:
\begin{equation*}
\begin{split}
d_{out,i}=((x_{UAV}-x_i)^2+(y_{UAV}-y_i)^2+(z_{UAV}-z_i)^2)^{0.5},~~~\forall z_{UAV}>z_i\\  
d_{out,i}=((x_{UAV}-x_i)^2+(y_{UAV}-y_i)^2+(z_i-z_{UAV})^2)^{0.5} ,~~~\forall z_{UAV}<z_i\\
\end{split}
\end{equation*}
Also, 
\begin{equation*}
\begin{split}
\theta_i=
sin^{-1}(\dfrac{(z_{UAV}-z_i)}{((x_{UAV}-x_i)^2+(y_{UAV}-y_i)^2+(z_{UAV}-z_i)^2)^{0.5}}),~~~\forall z_{UAV}>z_i\\ 
\theta_i=
sin^{-1}(\dfrac{(z_i-z_{UAV})}{((x_{UAV}-x_i)^2+(y_{UAV}-y_i)^2+(z_i-z_{UAV})^2)^{0.5}}),~~~\forall z_{UAV}<z_i\\
\end{split}
\end{equation*}
Rewrite the total path loss:
\begin{equation*}
\begin{split}
L_{Total}=
\sum_{i=1}^{m_1}\alpha_2log_{10}(d_{out,i})
+\dfrac{(\beta_2-\beta_1)}{(1+exp(-\beta_3(sin^{-1}(u)-\beta_4)))}~~~~~~~~~~~~~~~~~~~~~~~~~~~~~~~~~~~~~~~\\
+\sum_{i=1}^{m_2}\alpha_2log_{10}(d_{out,i})
+\dfrac{(\beta_2-\beta_1)}{(1+exp(-\beta_3(sin^{-1}(u)-\beta_4)))}+K~~~~~~~~~~~~~~~~~~~~~~~~~~~~~~~~~
\end{split}
\end{equation*}
Where:
\begin{equation*}
\begin{split}
u=(\dfrac{(z_{UAV}-z_i)}{((x_{UAV}-x_i)^2+(y_{UAV}-y_i)^2+(z_{UAV}-z_i)^2)^{0.5}}),~~~\forall z_{UAV}>z_i\\ 
u=(\dfrac{(z_i-z_{UAV})}{((x_{UAV}-x_i)^2+(y_{UAV}-y_i)^2+(z_i-z_{UAV})^2)^{0.5}}),~~~\forall z_{UAV}<z_i\\
K=\sum_{i=1}^{M}(\alpha_1+\alpha_3log_{10}f_{Ghz}+\beta_{1}+\gamma_{1}d_{in,i})~~~~~~~~~~~~~~~~~~~~~~~~~~~~~~~~~\\
\end{split}
\end{equation*}
Now, take the derivative with respect to $z_{UAV}$, we get:
\begin{equation*}
\begin{split}
\dfrac{dL_{Total}}{dz_{UAV}}=
\sum_{i=1}^{m_1}\dfrac{\alpha_2}{ln10}\frac{(z_{UAV}-z_i)}{((x_{UAV}-x_i)^2+(y_{UAV}-y_i)^2+(z_{UAV}-z_i)^2)}~~~~~~~~~~~~~~~~~~~~~~~~~~~~~~~~~~~~\\
+(\dfrac{-(\beta_2-\beta_1)(\frac{-\beta_3}{\sqrt{1-u^2}})(\frac{d_{out,i}-(z_{UAV}-z_i)^2d_{out,i}^{-1}}{d_{out,i}^2})}{(1+exp(-\beta_3(sin^{-1}u-\beta_4)))}.
\dfrac{exp(-\beta_3(sin^{-1}u-\beta_4))}{(1+exp(-\beta_3(sin^{-1}u-\beta_4)))})+~~~~~~~~~~~~~~~~~~~~~~~\\
\sum_{i=1}^{m_2}\dfrac{\alpha_2}{ln10}\frac{-(z_{i}-z_{UAV})}{((x_{UAV}-x_i)^2+(y_{UAV}-y_i)^2+(z_i-z_{UAV})^2)}~~~~~~~~~~~~~~~~~~~~~~~~~~~~~~~~~~~~~~~~~~~~~~\\
+(\dfrac{-(\beta_2-\beta_1)(\frac{-\beta_3}{\sqrt{1-u^2}})(\frac{-d_{out,i}+(z_{UAV}-z_i)^2d_{out,i}^{-1}}{d_{out,i}^2})}{(1+exp(-\beta_3(sin^{-1}u-\beta_4)))}.
\dfrac{exp(-\beta_3(sin^{-1}u-\beta_4))}{(1+exp(-\beta_3(sin^{-1}u-\beta_4)))})~~~~~~~~~~~~~~~~~~~~~~~~
\end{split}
\end{equation*}
The equation above equals zero when the UAV altitude equals the half of the building height, where the locations of indoor users are symmetric across the $xy$ and $xz$ planes.

	\bibliographystyle{IEEEtran}
	\bibliography{UAVpath}

\begin{thebibliography}{10}
\providecommand{\url}[1]{#1}
\csname url@samestyle\endcsname
\providecommand{\newblock}{\relax}
\providecommand{\bibinfo}[2]{#2}
\providecommand{\BIBentrySTDinterwordspacing}{\spaceskip=0pt\relax}
\providecommand{\BIBentryALTinterwordstretchfactor}{4}
\providecommand{\BIBentryALTinterwordspacing}{\spaceskip=\fontdimen2\font plus
\BIBentryALTinterwordstretchfactor\fontdimen3\font minus
  \fontdimen4\font\relax}
\providecommand{\BIBforeignlanguage}[2]{{%
\expandafter\ifx\csname l@#1\endcsname\relax
\typeout{** WARNING: IEEEtran.bst: No hyphenation pattern has been}%
\typeout{** loaded for the language `#1'. Using the pattern for}%
\typeout{** the default language instead.}%
\else
\language=\csname l@#1\endcsname
\fi
#2}}
\providecommand{\BIBdecl}{\relax}
\BIBdecl

\bibitem{haz2017efficient}
H.~Shakhatreh, A.~Khreishah, and B.~Ji, ``Providing wireless coverage to
  high-rise buildings using uavs,'' in \emph{(IEEE International Conference on
  Communications, IEEE ICC 2017 (accepted)}.\hskip 1em plus 0.5em minus
  0.4em\relax IEEE, 2017.

\bibitem{hazim2017}
H.~Shakhatreh, A.~Khreishah, A.~Alsarhan, I.~Khalil, A.~Sawalmeh, and
  O.~Noor~Shamsiah, ``Efficient 3d placement of a uav using particle swarm
  optimization,'' in \emph{The International Conference on Information and
  Communication Systems (ICICS 2017) (accepted)}.

\bibitem{bupe2015relief}
P.~Bupe, R.~Haddad, and F.~Rios-Gutierrez, ``Relief and emergency communication
  network based on an autonomous decentralized uav clustering network,'' in
  \emph{SoutheastCon 2015}.\hskip 1em plus 0.5em minus 0.4em\relax IEEE, 2015,
  pp. 1--8.

\bibitem{bor2016efficient}
R.~I. Bor-Yaliniz, A.~El-Keyi, and H.~Yanikomeroglu, ``Efficient 3-d placement
  of an aerial base station in next generation cellular networks,'' in
  \emph{Communications (ICC), 2016 IEEE International Conference on}.\hskip 1em
  plus 0.5em minus 0.4em\relax IEEE, 2016, pp. 1--5.

\bibitem{al2014modeling}
A.~Al-Hourani, S.~Kandeepan, and A.~Jamalipour, ``Modeling air-to-ground path
  loss for low altitude platforms in urban environments,'' in \emph{2014 IEEE
  Global Communications Conference}.\hskip 1em plus 0.5em minus 0.4em\relax
  IEEE, 2014, pp. 2898--2904.

\bibitem{mozaffari2015drone}
M.~Mozaffari, W.~Saad, M.~Bennis, and M.~Debbah, ``Drone small cells in the
  clouds: Design, deployment and performance analysis,'' in \emph{IEEE Global
  Communications Conference (GLOBECOM)}, 2015, pp. 1--6.

\bibitem{mozaffari2016optimal}
M.~Mozaffari\vspace{0mm}, W.~Saad, M.~Bennis, and M.~Debbah, ``Optimal
  transport theory for power-efficient deployment of unmanned aerial
  vehicles,'' \emph{IEEE International Conference on Communications (ICC),
  Kuala Lumpur, Malaysia,}, 2016.

\bibitem{mozaffari2016efficient}
M.~Mozaffari, W.~Saad, M.~Bennis, and M.~Debbah, ``Efficient deployment of
  multiple unmanned aerial vehicles for optimal wireless coverage,'' \emph{IEEE
  Communications Letters}, vol.~20, no.~8, pp. 1647--1650, 2016.

\bibitem{kalantari2016number}
E.~Kalantari, H.~Yanikomeroglu, and A.~Yongacoglu, ``On the number and 3d
  placement of drone base stations in wireless cellular networks,'' in
  \emph{IEEE Vehicular Technology Conference}, 2016, pp. 18--21.

\bibitem{ericsson}
``Ericsson report optimizing the indoor experience,
  http://www.ericsson.com/res/docs/2013/real-performance-indoors.pdf,'' 2013.

\bibitem{alcatel}
``In-building wireless: One size does not fit all,
  http://www.alcatel-lucent.com/solutions/in-building/in-building-infographic.''

\bibitem{cisco}
``Cisco service provider wi-fi: A platform for business innovation and revenue
  generation,
  http://www.cisco.com/c/en/us/solutions/collateral/service-provider/service-provider-wi-fi/solution{\_}overview{\_}c22{\_}642482.html.''

\bibitem{comm}
``Using high-power das in high-rise buildings,
  http://www.commscope.com/docs/using-high-power-das-in-high-rise-buildings-an-318376-ae.pdf.''

\bibitem{amplitic}
``Coverage solution for high-rise building,
  http://www.amplitec.net/products-2-coverage-solution-for-high-rise-building.html.''

\bibitem{zhang2016study}
S.~Zhang, Z.~Zhao, H.~Guan, and H.~Yang, ``Study on mobile data offloading in
  high rise building scenario,'' in \emph{Vehicular Technology Conference (VTC
  Spring), 2016 IEEE 83rd}.\hskip 1em plus 0.5em minus 0.4em\relax IEEE, 2016,
  pp. 1--5.

\bibitem{series2009guidelines}
M.~Series, ``Guidelines for evaluation of radio interface technologies for
  imt-advanced,'' \emph{Report ITU}, no. 2135-1, 2009.

\bibitem{imai2016outdoor}
T.~Imai, K.~Kitao, N.~Tran, N.~Omaki, Y.~Okumura, and K.~Nishimori,
  ``Outdoor-to-indoor path loss modeling for 0.8 to 37 ghz band,'' in
  \emph{Antennas and Propagation (EuCAP), 2016 10th European Conference
  on}.\hskip 1em plus 0.5em minus 0.4em\relax IEEE, 2016, pp. 1--4.

\bibitem{sutton1998reinforcement}
R.~S. Sutton and A.~G. Barto, \emph{Reinforcement learning: An
  introduction}.\hskip 1em plus 0.5em minus 0.4em\relax MIT press Cambridge,
  1998, vol.~1, no.~1.

\bibitem{kennedy1995particle}
J.~Kennedy and R.~Eberhart, ``Particle swarm optimization,'' in \emph{Neural
  Networks, 1995. Proceedings., IEEE International Conference on},
  vol.~4.\hskip 1em plus 0.5em minus 0.4em\relax IEEE, 1995, pp. 1942--1948.

\bibitem{clerc2002particle}
M.~Clerc and J.~Kennedy, ``The particle swarm-explosion, stability, and
  convergence in a multidimensional complex space,'' \emph{IEEE transactions on
  Evolutionary Computation}, vol.~6, no.~1, pp. 58--73, 2002.

\bibitem{korf2002new}
R.~E. Korf, ``A new algorithm for optimal bin packing,'' in \emph{AAAI/IAAI},
  2002, pp. 731--736.

\bibitem{ng2000cs229}
A.~Ng, ``Cs229 lecture notes.''

\end{thebibliography}
	
	\end{document}